\documentclass[12pt]{iopart}
\usepackage{iopams}
 \usepackage{mathrsfs}
 \usepackage{amsfonts,amsthm}
\usepackage[all]{xy}

\newtheorem{theorem}{Theorem}
\newtheorem{lemma}{Lemma}
\newtheorem{prop}{Proposition}
\theoremstyle{remark}\newtheorem{remark}{Remark}
				\newtheorem{corollary}{Corollary}
\theoremstyle{definition}\newtheorem{definition}{Definition}

\newcommand{\RR}{{\mathbb R}}
\newcommand{\q}{q}
\newcommand{\aaa}{\tau}
\newcommand{\NN}{{\mathbb N}}

\newcommand{\QQ}{{\mathbb Q}}
\newcommand{\LLL}{{\mathcal  L}}
\newcommand{\DDD}{{\mathcal  D}}
\newcommand{\NNN}{{\mathcal  N}}
\newcommand{\OOO}[1]{{\mathcal{O}\left(#1\right)}}

\begin{document}
\title[Short title]{Inverse scattering for star-shaped nonuniform lossless electrical networks}
\author{Filippo Visco-Comandini, Mazyar Mirrahimi and Michel Sorine}
\address{INRIA Paris-Rocquencourt, Domaine de Voluceau, Rocquencourt, B.P 105 -78153 Le Chesnay Cedex, Fr}
\eads{\mailto{filippo.comandini@inria.fr}, \mailto{mazyar.mirrahimi@inria.fr}, \mailto{michel.sorine@inria.fr}}
\begin{abstract}
The Frequency Domain Reflectometry (FDR) is studied as a powerful tool to detect hard or soft faults in star-shaped networks  of nonuniform lossless transmission lines. Processing the FDR measurements leads to solve an inverse scattering problem for a Schr\"odinger operator on a star-shaped graph. Throughout this paper, we restrict ourselves to the case of minimal experimental setup corresponding to only one diagnostic port plug. First, by studying the asymptotic behavior of the reflection coefficient in the high-frequency limit, we prove the identifiability of the geometry of this star-shaped graph: the number of edges and their lengths. The proof being rather constructive, it provides a method to detect the hard faults in the network. Next, we study the potential identification problem by inverse scattering, noting that the potentials represent the inhomogeneities due to the soft faults in the network wirings. Here, the main result states that the measurement of two reflection coefficients, associated to two
different sets of boundary conditions at the extremities of the tree, determines uniquely the potentials; it is a generalization of the theorem of the two boundary spectra on an interval~\cite{Borg-46}. \\

\end{abstract}


\noindent{\it Inverse scattering, Telegrapher's equation, Inverse Sturm-Liouville problem\/}


\maketitle
\section{Introduction}\label{sec:intro}


The rather extensive literature concerning the ``inverse scattering problem'' and the ``inverse Sturm-Liouville problem'' on graphs have mostly followed separate pathways except for a very few results~\cite{Gerasimenko-87,Gerasimenko-88}. In the following paragraphs, we briefly recall the previous results on these subjects and at the end we will situate the result of this paper with respect to the others. Indeed, as it will be seen later, the inverse Sturm-Liouville problem considered in this paper raises from the necessity of finding a minimal setup for solving the inverse scattering problem.

The paper~\cite{harmer-02} considers a star-shape graph consisting of $N$ infinite branches and solves the inverse scattering problem assuming the measurement of $N-1$ reflection coefficients. Next, in the paper~\cite{harmer-05}, Harmer provides an extension of the previous result  with general self-adjoint boundary conditions at the central node. This however necessitates the knowledge of  $N$ reflection coefficients.\\
The paper~\cite{kurasov-stenberg} studies the relation between the scattering data and the topology of the graph. They show that the knowledge of the scattering matrix is not enough to determine uniquely the topological structure of a generic graph.
In ~\cite{avdonin-kurasov-08}, Avdonin and Kurasov consider again a star-shape graph with $N$ finite branches. They, show that the knowledge of one diagonal element of the response operator allows one to reconstruct the potential on the edge corresponding to this element.

As mentioned above, in parallel to the research on inverse scattering problems, another class of papers consider the inverse problem for Sturm-Liouville operators on compact graphs.
These results can be seen as extensions of the classical result provided by Borg~\cite{Borg-46}, on the recovering of the Sturm-Liouville operator from two spectra on a finite interval.
The main progress in this field has been made by Yurko~\cite{yurko-05},~\cite{yurko-07} and~\cite{yurko-08}.
The paper~\cite{yurko-05} deals with the inverse spectral problem on a tree. The idea is to generalize the Borg's result in the following sense: for a tree with $n$ boundary vertices, it is sufficient to know $n$ spectra, corresponding to $n$ different settings for boundary conditions at the extremities, to retrieve the potentials on the tree. In the recent paper~\cite{yurko-08}, the same kind of result  is proposed for a star-shape graph including a loop joined to the central node. Finally,  \cite{yurko-07} provides  a generalization of~\cite{yurko-05} to higher  order differential operators on a star-shape graph. \\
In~\cite{pivovarchik-07}, the author proves that under some restrictive assumptions on the spectrum of a Sturm-Liouville operator on a star-shape graph with some fixed boundary conditions, the knowledge of this spectra can determine uniquely the Sturm-Liouille operator.

In this paper, we consider the inverse scattering problem motivated by the application in fault-detection/diagnostic of star-shape LC transmission networks. We are interested in minimal experimental setting providing enough information on the network and the potentials on its branches. The graph consists of $N$ finite branches joined at a central node and we add a infinite branch to this central node for the experimentation. We will see that the knowledge of the reflection coefficient is equivalent to the knowledge of the spectra for the Sturm-Liouville operator defined on the compact part of the graph for various boundary conditions at the central node. By considering the associated inverse Sturm-Liouville problem, we will show that under some assumptions on the geometry of the metric graph, the knowledge of only two reflection coefficients, corresponding to two settings for the boundary conditions at the terminal nodes, is enough to determine the potentials (at least locally).

In the next section, we explain the application under study and the associated experimental setup. We will show that the inverse scattering problem for the so-called Telegrapher's equation (lossless case) on the network is equivalent to an inverse scattering problem for a Schr\"odinger operator over the metric graph of the network. In Section~\ref{sec:direct}, we will consider the direct scattering problem and we will characterize the reflection coefficient in terms of the fundamental solutions for Sturm-Liouville operators on branches. In Section~\ref{sec:geometry}, we will show that the knowledge of only one reflection coefficient is enough to identify the lengthes of the branches of the metric graph. This result will be useful to locate hard faults (open or short circuits). Finally, in Section~\ref{sec:potential} we consider the main problem of recovering the potentials from the knowledge of one or two reflection coefficients.  This result will be useful to locate soft faults (local variations of the electrical characteristics).We will prove the equivalence of the inverse scattering problem with an associated inverse Sturm-Liouville problem on the compact part of the graph. This inverse problem is then treated applying the methods extending the classical result by Borg~\cite{Borg-46}.

\section{Frequency domain reflectometry}\label{sec:fdr}
The electric signal transmission through a wired network is, generally, modeled  with  the ``Telegrapher's equation'' and characterized by the parameters $L, C, R, G$ (functions of the space position $z$ along the transmission lines) representing the inductance $L$, capacitance $C$, resistance $R$  and loss conductance $G$ per unit of length.  These parameters allow a rather complete and understandable description of the transmission lines  and are sufficient to represent the lines in the frequency range used during reflectometry. In the sequel we will suppose that this model can be used for all frequencies. However, it appears to be impossible to retrieve all these parameters uniquely through the information provided by reflectometry experiments. Everywhere, through this paper, we will  consider the simpler nonuniform lossless situation  ($R=G=0$). As we will see later, the reflectometry measurement is  still not enough to retrieve the both parameters $L$ and $C$ but rather an aggregate of these two parameters, the local characteristic impedance $Z_c(z):=\sqrt{L(z)/C(z)}$.

Following~\cite{Kay-72} and~\cite{Jaulent-82}, the presentation of the reflectometry experiment of this section, is based on a model derived from the ``Telegrapher's equation'' and parameterized by $Z_c$. To cope with the network case, we have translated the Kirchhoff  rules at the nodes of the network within this new modeling framework. Note that, in this paper, we restrict ourselves to the case of a simple star-shape network and therefore the only node of the graph where the Kirchhoff rules need to be adapted is the central one. The  faults, in which we are interested through this approach, are represented by the lengths of the branches (hard faults) and by the heterogeneities of $Z_c$ along the branches (soft faults). The considered reflectometry experiment model is based on a far-field method consisting in adding a uniform infinite wire joined to the network at its central node. In practice, connecting  a matched charge to the extremity of a finite line,  is sufficient to emulate the electrical propagation through an infinite line.

The linearity of the transmission line model allows to replace any test by an equivalent test in harmonic regime. We can therefore start by stating the Telegrapher's equation in the harmonic regime,  i.e. the tension and the current intensity are respectively of the form $e^{ - \imath \omega t}V(\omega,z) $ and $e^{ - \imath \omega t}I(\omega,z)$, where $\omega$ is the time frequency and $z$ the position. On each line, we have

\begin{equation}
\eqalign{
\frac{\partial}{\partial z} V(\omega,z) - \imath \omega L(z)I(\omega,z)=0,\cr
\frac{\partial}{\partial z} I(\omega,z) - \imath \omega C(z)V(\omega,z)=0,}\label{eq:telegraph}
\end{equation}

Through this paper, we assume that
\begin{description}
  \item[\textbf{A1}] the distributed parameters $C(z)$ and $L(z)$ are twice continuously differentiable on the transmission lines;
  \item[\textbf{A2}] they are strictly positive, $C(z)>0,L(z)>0$;
  \item[\textbf{A3}] the characteristic impedance $Z_c(z):=\sqrt{L(z)/C(z)}$ is continuous at the central node of the star-shape network;
  \item[\textbf{A4}] the transmission lines are uniform in a neighborhood of the extremities of the branches.
\end{description}

\medskip

{\em The Liouville transformation. } Note that the reflectometry experiment leads to observing the tensions and currents along the time at some position: only the travelling times (and amplitudes) of waves are accessible by such experiments. A fault can only be localized in terms of the traveling time of  the reflected test wave  starting from the test point. This leads to a particular change of variables,  the Liouville transformation, allowing to work with the traveling time rather than spatial coordinates.
Let us recall this transformation:
$$
x(z)=\int_0^z\sqrt{L(z)C(z)}ds
$$
which corresponds to the wave traveling time from the position $0$ to the position $z$.
Remark that after this transformation, $\omega$ is also the wave number on each branch.

The inverse transformation being well defined, we will write $C(x)\equiv C(z(x))$,  $L(x)\equiv L(z(x))$, $V(\omega,x) \equiv V(\omega,z(x))$ and $I(\omega,x) \equiv I(\omega,z(x))$.

The Telegrapher's equation~\eref{eq:telegraph} becomes

\begin{equation}
\eqalign{
\frac{\partial}{\partial x} V(\omega,x) - \imath \omega Z_c(x)I(\omega,x)=0,\cr
\frac{\partial}{\partial x} I(\omega,x) - \imath \omega Z_c(x)^{-1} V(\omega,x)=0. }\label{eq:telegraph-x}
\end{equation}

\medskip

{\em The wave decomposition and equivalent forms of the Telegrapher's equation. } Define $ V_{\pm} (\omega,x) = \frac{1}{2} (V(\omega,x) \pm Z_c(x)I(\omega,x)) $ and
$Q(x) = \displaystyle \frac{1}{2}Z_c(x)^{-1} \frac{d Z_c(x)}{dx}$.
We have the following decomposition of $V$:

\begin{equation}
\eqalign{
V(\omega,x) = V_{+} (\omega,x) + V_{-} (\omega,x), \cr
\frac{\partial}{\partial x} V_{\pm} (\omega,x)  \mp \imath \omega V_{\pm} (\omega,x)  =
\pm Q(x) (V_{+} (\omega,x) - V_{-} (\omega,x)) }\label{eq:telegraph-Vpm}
\end{equation}

In particular, in an interval where a branch is uniform, $ Q(x) =0$, and the solution is the sum of waves of opposite directions. For any $\bar x$ and $x$ in this interval:
$$
V(\omega,x) = V_{+}(\omega, \bar x) e^{  \imath \omega (x - \bar x)} + V_{-}(\omega, \bar x)  e^{ - \imath \omega (x - \bar x)}.
$$

\medskip

Define now $y(\omega,x)=\sqrt{Y_c(x)} V(\omega,x)$ and  $\q (x)=\sqrt{Z_c(x)}\frac{d^{2}}{dx^{2}}\sqrt{Y_c(x)}$, with $Y_c(x)=Z_c(x)^{-1}$.
The Telegrapher's equation~\eref{eq:telegraph} becomes a Schr\"odinger equation:
\begin{equation}
-\frac{d^2 y}{dx^2}(\omega,x)+q(x) y(\omega,x)=\omega^2y(\omega,x). \label{eq:schrodinger}
\end{equation}

 It can be seen from \eref{eq:telegraph-Vpm} and \eref{eq:schrodinger}, that, the knowledge of the potential $Q(x)$ or of $q(x)$ and of the boundary conditions on $V$ and $I$ is sufficient to compute the solution of \eref{eq:telegraph} on the network.
In our lossless situation, we have chosen $q$ as the parameter to be identified through the reflectometry experiment. A variant of \eref{eq:telegraph-Vpm}, the Zakharov-Shabat equations, would be the good choice in the more general lossy case. Remark that
$   \displaystyle \frac{dQ}{dx} - Q^2 + q = 0$.

\medskip

{\em The reflection coefficient.} Taking $e^{  \imath \omega (x - t)}$, oriented toward the increasing $x$, as a reference forward wave, the reflection coefficient is the following ratio of backward over forward wave amplitudes:
$R(\omega, x) = \frac{ V_{-}(\omega, x)  e^{ \imath \omega x}}{V_{+}(\omega, x)  e^{ - \imath \omega x}} =
e^{ 2 \imath \omega x} \frac{ V_{-}(\omega, x) }{V_{+}(\omega, x)} $.  In particular, $R(\omega, x) = R(\omega, \bar x)$: defined in this way, the reflection coefficient is constant in intervals where $Z_c$ is constant.
For an arbitrary $x$, denoting by $Z(\omega,x)=V(\omega,x)/I(\omega,x)$ the (possibly infinite) apparent impedance at $x$, we still define the reflection coefficient as being
\begin{equation}
R(\omega, x) = e^{ 2 \imath \omega x}\frac{Z(\omega,x) - Z_c(x)}{Z(\omega,x) + Z_c(x)}.
  \label{eq:R-ingenieur}
\end{equation}
With this definition, in general, $V_{-}(\omega, x) = e^{ -2 \imath \omega x} R(\omega, x) V_{+} (\omega,x)$, and, using \eref{eq:telegraph-Vpm}, it is easy to check that $R$ is solution of the following Riccati equation:
\begin{equation}
  \displaystyle \frac{\partial R(\omega, x)}{\partial x} - e^{ -2 \imath \omega x} Q(x) R(\omega, x)^2 + Q(x)e^{ 2 \imath \omega x} = 0.
   \label{eq:ricc-R}
\end{equation}
Finally, note that, we will only consider ``positive real'' terminal impedances $Z(\omega, \tau )$, in the sense that $Z(- \omega, \tau ) = \bar{Z}(\omega, \tau )$ and $\Re Z(\omega, \tau ) \geq 0$. This together with the fact that $Z_c(x)\in \RR$ is positive implies that $|R(\omega,x)|<1$. Furthermore, experiments with $\omega \geq 0$ are sufficient.

\medskip

{\em The network under test. }
Throughout this paper, $\Gamma$ represent the compact star-shape network consisting of the branches $(e_j)_{j=1}^N$ joining at the central node and $\Gamma^+$ is the extended graph where the test branch $e_0$ is also added to the graph.
We have $N+1$ equations of the form
\begin{equation}
-\frac{d^{2}y_{j}}{dx^{2}}+\q _{j}(x)y_{j}=\omega^2y_j \qquad x\in(0,\aaa_j), \label{eq:schrodinger-j}
\end{equation}
where $\tau_j$ is the wave traveling time associated to the branch number $j$ ($\aaa_0=\infty$ as the added branch $e_0$ is assumed to be an infinite line).
In particular note that, as the infinite branch $e_0$ is assumed to be a uniform transmission line, we have $\q_0(x)=0,\qquad x\in(0,\infty)$.

\medskip

{\em The boundary condition for the reflectometer.} Consider now that a generator with a matched internal impedance $Z_c$ and an electromotive force $2 V_g e^{ \imath \omega \bar x}$ is connected in $\bar x$ of an interval $\mathcal{I}$ where $Z_c$ is constant. We have  $V(\omega, \bar  x) + Z_c(\bar  x)I(\omega, \bar  x) = 2 V_g e^{ \imath \omega \bar x} $ which can be the boundary condition if the branch terminates at $\bar x$.
We have $V_{+} (\omega, \bar x) = V_g e^{ \imath \omega \bar x} $ and
$V(\omega, x) = V_g (e^{ \imath \omega x} + R(\omega,  x)  e^{ - \imath \omega  x}) $ for all $x \in\mathcal{I}$ . As,  $R(\omega, x) = R(\omega, \bar x)$, the reflection coefficient can then be determined from the measurement of the tension $V(\omega, x)$ anywhere in $\mathcal{I}$.
In the sequel, we will use a test branch  $e_0$ with a constant $Z_c$ connecting a matched generator to the central node . The measured  reflection coefficient on this branch, will be simply written $R(\omega)$.
As we will work with $y$, we choose $V_g = \sqrt{Z_c(x)}$, so that $y(\omega, x) = e^{ \imath \omega x} + R(\omega)  e^{ - \imath \omega  x}$.
Finally, it will be convenient to take the same positive orientation on all the branches, from the central node at $x=0$ toward the increasing $x$. Our reference forward wave on $e_0$ is then in the direction of the decreasing $x$, so that, changing $x$ into $-x$, and supposing $e_0$ of infinite length, the boundary condition for the reflectometer is:
\begin{equation}
 y(\omega,x)\sim e^{-\imath\omega x}+R(\omega)e^{\imath\omega x}\qquad \textrm{as }x\rightarrow\infty \quad \textrm{in } e_0.
\end{equation}

\medskip

{\em The two sets of boundary conditions at the network extremities.}
In order to recover the potential of the star-shape network, we will need to consider two experimental settings, with open circuit or short circuit at the extremities of the branches. This will lead to a problem  similar to solving an  inverse spectral problem for the Sturm-Liouville operator when two spectra are known.

The first setting corresponds to open circuit configuration at the extremities of the finite branches ($(e_j)_{j=1}^N$). This, together with the Assumption~\textbf{A4} on the local uniformity of the lines around $\tau_j$'s, leads to boundary conditions of the form $I_j(\omega,\tau_j)=0$, or equivalently, we obtain the setting called, the Neumann configuration:
\begin{equation}\label{eq:bcN}
y'_j(\tau_j)=0~~ \qquad j=1,\cdots,N.
\end{equation}

The second setting corresponds to the short circuit configuration at the extremities of the finite branches ($(e_j)_{j=1}^N$). This leads to boundary conditions of the form $V_j(\omega,\tau_j)=0$, or equivalently, we obtain the setting called, the Dirichlet configuration:
\begin{equation}\label{eq:bcD}
y_j(\tau_j)=0~~ \qquad j=1,\cdots,N.
\end{equation}

\begin{remark}\label{rem:riccati}
In some of the applications that we have in mind, the reflectometry experiment has to take place without perturbing significantly the normal utilization of the transmission network,  so that using open or short circuits conditions would be impossible. There is a way to circumvent this problem by computing the results of the open or short circuit experiments from results of two less invasive experiments.
The idea is to use nonlinear superposition properties of solutions of Riccati equations as in \cite{sorine-riccati}, in order to get a closed-form  representation of the reflection coefficient, solution of \eref{eq:ricc-R}, as a function of a general load impedance (value of $Z$ at the extremity of a branch) and of two particular solutions corresponding to two load impedances more compatible with the network utilization. This will be presented in a forthcoming paper.
\end{remark}

\medskip

{\em The boundary condition at the central node.} It writes
\begin{equation}
\eqalign{y_{i}(\omega,0)=y_j(\omega,0)=:\bar y(\omega)\qquad\qquad
i,j=0,\cdots,N,\cr
\sum_{j=1}^N y'_j(\omega,0)-y'_0(\omega,0)=-\frac{1}{2}\frac{ \sum_{j=1}^N (Z_c^j)'(0) }{Z_{c}^0} \bar y(\omega),}\label{eq:kirchhoff2}
\end{equation}
where $y'_j(\omega,0)$ and $(Z_c^j)'(0)$ denote the spatial derivatives at the point $x=0$ and $Z_c^j$ is the characteristic impedance of the branch number $j$. Note, in particular, that we have applied the continuity of $Z_c^j$'s at the central node (Assumption \textbf{A3}): $Z_c^j(0)=Z_c^0$, $\forall j$.

\medskip

{\em Formulation of the model. } In conclusion, in order to study the $LC$-transmission line
equations on the graph $\Gamma^+$, we can study the Schr\"odinger operators
\begin{eqnarray}
\LLL^+_{\NNN,\DDD}&=\otimes_{j=0}^N (-\frac{d^2}{dx^2} +q_j(x)),\nonumber\\
D(\LLL^+_{\NNN,\DDD})&= \textrm{closure of } C_{\NNN,\DDD}^\infty \textrm{ in }
H^2(\Gamma^+),\label{eq:schN}
\end{eqnarray}
where $C_{\NNN}^\infty(\Gamma^+)$ (resp. $C_{\DDD}^\infty(\Gamma^+)$) denotes the space of infinitely
differentiable functions $f=\otimes_{j=0}^N f_j$ defined on $\Gamma^+$ satisfying the boundary
conditions
\begin{eqnarray}
f_j(0)=f_{j'}(0) \qquad&j,j'=0,\cdots,N,\nonumber\\
\sum_{j=1}^N f_j'(0)- f_0'(0)= H f_0(0),\qquad &H=-\frac{1}{2}\frac{\left(\sum_{j=1}^N (Z^j_c)'(0)\right)}{Z_c^0},\label{eq:H}\\
f_j'(\tau_j)=0\quad (f_j(\tau_j)=0 \textrm{ for } C_{\DDD}^\infty(\Gamma^+)),\qquad  & j=1,\cdots,N. \nonumber
\end{eqnarray}

\section{Direct scattering problem}\label{sec:direct}
The operators $(\LLL^+_{\NNN,\DDD},D(\LLL^+_{\NNN,\DDD}))$ are essentially self-adjoint. To prove this fact
we observe first that these operators are a compact perturbation
of the operators  $\otimes_{j=0}^n\left(-\frac{d^2}{dx^2}\right)$ with
the same boundary conditions.
Now, we apply a general result by Carlson~\cite{carlson-98} on the
self-adjointness of differential operators on graphs.
Indeed, following the Theorem 3.4 of~\cite{carlson-98}, we only need
to show that at a node connecting $m$ edges, we have $m$ linearly
independent linear boundary conditions. At the terminal nodes of
$\{e_j\}_{j=1}^{N}$ this is trivially the case as there is one
branch and one boundary condition (Dirichlet or Neumann). At the central node it is not
hard to verify that~\eref{eq:kirchhoff2} define $N+1$ linearly independent boundary
conditions as well. This implies that the operators
$(\LLL^+_{\NNN,\DDD},D(\LLL^+_{\NNN,\DDD}))$ are essentially self-adjoint and therefore that they admit a
unique self-adjoint extension on $L^2(\Gamma^+)$. \\
We are interested in the scattering solution where a signal of
frequency $\omega$ is applied at the infinite extremity of the infinite branch.
In such a case, we will be
seeking a solution satisfying the asymptotic behavior
$$
y_0(x,\omega) \sim e^{-\imath \omega x}+R(\omega) e^{\imath \omega
x},\qquad \textrm{for } x\rightarrow\infty.
$$
\begin{prop}\label{prop:scattering}
Under the assumptions \textbf{A1} through \textbf{A4}, there exists a unique solution
$$
\Psi(x,\omega)=\otimes_{j=0}^N y_j(x,\omega),
$$
continuous with respect to $\omega$, of the scattering problem, satisfying
\begin{itemize}
    \item $ -\frac{d^2}{dx^2} y_j(x,\omega)+q_j(x) y_j(x,\omega)=\omega^2 y_j(x,\omega)$ for $j=0,\cdots, N$;
    \item $\left(y_j(x,\omega)\right)_{j=0}^{N}$ satisfy the boundary
    conditions~\eref{eq:kirchhoff2} and~\eref{eq:bcN} or~\eref{eq:bcD} ;
    \item For each $\omega\in\RR$, there exist $R(\omega)$
    such that
    \begin{equation}\label{eq:req1}
    y_0(x,\omega) \sim e^{-i \omega x}+R(\omega) e^{i \omega x},\quad
    x\rightarrow \infty.
    \end{equation}
\end{itemize}
We will  denote the reflection coefficient $R(\omega)$ defined by~\eref{eq:req1}  in the Neumann (resp. Dirichlet) case by $R_{\NNN}(\omega)$ (resp. $R_{\DDD}(\omega)$). This coefficient appears to be unique.
\end{prop}
\begin{proof}
This proof gives us a concrete method for obtaining
scattering solutions. Indeed, we will propose a solution and we will show that it is the unique one. \\
In this aim, we need to use Dirichlet/Neumann fundamental solutions of a Sturm-Liouville boundary problem.
	\begin{definition}\label{def:fund}
    Consider the potentials $q_j$ as before and extend them by 0 on $(-\infty,0)$ so that they are defined
    on the intervals $(-\infty,\tau_j]$.
    The Dirichlet (resp. Neumann) fundamental solution $\varphi^j_{\DDD}(x,\omega)$ (resp. $\varphi^j_{\NNN}(x,\omega)$), is a solution of the equation,
	\begin{eqnarray*}
	-\frac{d^2}{dx^2}\varphi^j_{\DDD,\NNN}(x,\omega)+ \q_j (x)\varphi^j_{\DDD,\NNN}(x,\omega)
	=\omega^2 \varphi^j_{\DDD,\NNN}(x,\omega),\qquad x\in(-\infty,\tau_j),\\
	\varphi^j_{\DDD}(\tau_j,\omega)=0,\qquad (\varphi^j)'_{\DDD}(\tau_j,\omega)=1,\\
    \varphi^j_{\NNN}(\tau_j,\omega)=1,\qquad (\varphi^j)'_{\NNN}(\tau_j,\omega)=0.
	\end{eqnarray*}
	\end{definition}
Consider, now,  the function
$$
\Psi_{\DDD,\NNN}(x,\omega)=\otimes_{j=0}^N \Psi^j_{\DDD,\NNN}(x,\omega),
$$
where
\begin{eqnarray*}
\Psi^0_{\DDD,\NNN}(x,\omega)&=e^{-\imath \omega x} +R_{\DDD,\NNN}(\omega) e^{\imath \omega x},\qquad &x\in[0,\infty),\nonumber\\
\Psi^j_{\DDD,\NNN}(x,\omega)&=\alpha^j_{\DDD,\NNN}(\omega) \varphi^j_{\DDD,\NNN}(x,\omega),\qquad &x\in[0,\tau_j],~~j=1,\cdots,N.
\end{eqnarray*}

Here the coefficients $R_{\DDD,\NNN}$ and $\alpha^j_{\DDD,\NNN}$ are given by the boundary conditions~\eref{eq:kirchhoff2} at the central node: \small
\begin{eqnarray}
1+R_{\DDD,\NNN}(\omega)=\alpha^j_{\DDD,\NNN}(\omega)\varphi^j_{\DDD,\NNN}(0,\omega),\qquad j=1,\cdots,N,
\label{eq:R1}\\
\fl\qquad\qquad\sum_{j=1}^N\alpha^j_{\DDD,\NNN}(\omega)(\varphi_{\DDD,\NNN}^{j})'(0,\omega)+\imath \omega(1-R_{\DDD,\NNN}(\omega))
=H(1+R_{\DDD,\NNN}(\omega)).\label{eq:R2}
\end{eqnarray}\normalsize
One, easily, sees that this $\Psi_{\DDD,\NNN}$ is in $D(\LLL^+_{\NNN,\DDD})$, the domain of the operator, and satisfies the conditions of the proposition as soon as the coefficients $R_{\DDD,\NNN}$ and $\left(\alpha_{\DDD,\NNN}^j\right)_{j=1}^N$ are continuous with respect to $\omega$. This, trivially, provides the existence of a scattering solution. Here, we show that $\Psi_{\DDD,\NNN}$ is actually the unique one.

Assume that there exists another $Y_{\DDD,\NNN}=\otimes_{j=0}^N Y^j_{\DDD,\NNN}(x,\omega)$ solution of the scattering problem. By the compact injection of $H^2$ in $C^1$, we now that $Y^j_{\DDD,\NNN}(x,\omega)$ and $\Psi^j_{\DDD,\NNN}(x,\omega)$ are $C^1$ functions of $x$ over $[0,\tau_j]$. Therefore, the Wronskian
$$
W(Y^j_{\DDD,\NNN}(.,\omega),\Psi^j_{\DDD,\NNN}(.,\omega))=Y^j_{\DDD,\NNN}(.,\omega)(\Psi^j_{\DDD,\NNN})'(.,\omega)-
\Psi^j_{\DDD,\NNN}(.,\omega)(Y^j_{\DDD,\NNN})'(.,\omega)
$$
is well-defined. Moreover, as the potentials $\q_j(x)$ are continuous functions over $[0,\tau_j]$ and as
$Y^j_{\DDD,\NNN}(.,\omega)$ and $\Psi^j_{\DDD,\NNN}(.,\omega)$ are solutions of the associated Sturm-Liouville equation, they are in fact $C^2$ over $[0,\tau_j]$. Thus, the derivative of the Wronskian is also well defined over $[0,\tau_j]$. Through a simple computation and by noting that $Y^j_{\DDD,\NNN}(.,\omega)$ and $\Psi^j_{\DDD,\NNN}(.,\omega)$ are solutions of the same Sturm-Liouville equations, one has
$$
\frac{d}{dx} W(Y^j_{\DDD,\NNN}(x,\omega),\Psi^j_{\DDD,\NNN}(x,\omega))=0,\qquad x\in[0,\tau_j],
$$
and so the Wronskian remains constant over the interval $[0,\tau_j]$.

For the finite branches $(e_j)_{j=1}^N$, applying the (Dirichlet or Neumann) boundary conditions at the terminal nodes, we easily have
$$
W(Y^j_{\DDD,\NNN}(\tau_j,\omega),\Psi^j_{\DDD,\NNN}(\tau_j,\omega))= 0,
$$
and therefore the Wronskian is identically 0 over the whole branch. This implies that, $Y^j_{\DDD,\NNN}(.,\omega)$
and $\Psi^j_{\DDD,\NNN}(.,\omega)$ are co-linear:
$$
Y^j_{\DDD,\NNN}(x,\omega)=\beta^j_{\DDD,\NNN}(\omega)\varphi_{\DDD,\NNN}^j(x,\omega),\qquad x\in[0,\tau_j],\quad j=1,\cdots,N.
$$
Over the branch $e_0$, as $Y^0_{\DDD,\NNN}(.,\omega)$ satisfies a homogenous Sturm-Liouville equation ($q_0=0$), it necessarily admits the following form
$$
Y^0_{\DDD,\NNN}(x,\omega)=e^{-\imath \omega x}+\widetilde R_{\DDD,\NNN}(\omega) e^{\imath \omega x}.
$$
What remains to be shown is that one necessarily has $\widetilde R_{\DDD,\NNN}(\omega)\equiv R_{\DDD,\NNN}(\omega)$ and
similarly $\beta^j_{\DDD,\NNN}(\omega)\equiv \alpha^j_{\DDD,\NNN}(\omega)$.

Indeed, the equations~\eref{eq:R1} and~\eref{eq:R2} provide $N+1$ linear relations for the $N+1$ unknown coefficients $R_{\DDD,\NNN}$ and $\left(\alpha_{\DDD,\NNN}^j\right)_{j=1}^N$. Trivially, as soon as, the coefficients $\left(\varphi^j_{\DDD,\NNN}(0,\omega)\right)_{j=1}^N$ are non-zero, these linear relations are independent and there exists a unique solution for the unknowns $R_{\DDD,\NNN}$ and $\left(\alpha_{\DDD,\NNN}^j\right)_{j=1}^N$. However, the zeros of each one of the coefficients $\left(\varphi^j_{\DDD,\NNN}(0,\omega)\right)_{j=1}^N$ correspond to isolated values of $\omega$ (square-root of the eigenvalues of the operator $-\frac{d^2}{dx^2}+q_j(x)$ with Dirichlet boundary condition at $x=0$ and Dirichlet or Neumann boundary condition at $x=\tau_j$). Therefore, the coefficients $R_{\DDD,\NNN}$ and $\left(\alpha_{\DDD,\NNN}^j\right)_{j=1}^N$ are well-defined except for a set of isolated values of $\omega$. However, as these coefficients need to be continuous with respect to $\omega$, they will, also, be defined uniquely over these singular points.

Furthermore, dividing~\eref{eq:R2} by $(1+R_{\DDD,\NNN}(\omega))$ and inserting~\eref{eq:R1}, we find the explicit formula
\begin{equation}\label{eq:R}
\frac{1-R_{\DDD,\NNN}(\omega)}{1+R_{\DDD,\NNN}(\omega)}=\frac{H}{\imath\omega}-\frac{1}{\imath\omega}\sum_{j=1}^N \frac{(\varphi_{\DDD,\NNN}^j)'(0,\omega)}{\varphi_{\DDD,\NNN}^j(0,\omega)}.
\end{equation}
Finally, inserting the value of $R_{\DDD,\NNN}(\omega)$ found in~\eref{eq:R} into~\eref{eq:R1}, we find
\begin{equation*}
\alpha_{\DDD,\NNN}^j(\omega)=\frac{1+R_{\DDD,\NNN}(\omega)}{\varphi_{\DDD,\NNN}^j(0,\omega)}.
\end{equation*}
\end{proof}
\section{Inverse scattering for geometry identification}\label{sec:geometry}
As a first inverse problem, we consider the inversion of the geometry of the network. In fact, we will prove the well-posedness of the inverse problem of finding the number of branches $N$ and the  lengthes $(\tau_j)_{j=1}^N$ of a star-shape graph through only one reflection coefficient $R_\NNN(\omega)$ (the case of Dirichlet reflection coefficient can be treated similarly). Furthermore, as we will see through the proof of the Theorem~\ref{thm:geometry}, the method is rather constructive and one can think of an algorithm to identify the lengthes, at least approximately. The proof is based on an asymptotic analysis in high-frequency regime of the reflection coefficient and some classical results from the theory of almost periodic functions (in Bohr sense). Before, announcing the main Theorem, we need a few lemmas. A first lemma precises the high frequency behavior of the Neumann fundamental solutions $(\varphi^j_\NNN)_{j=1}^N$:
\begin{lemma}\label{lem:fund}
Consider a potential $q$ in $C^0((-\infty,\tau])$ and take the Neumann fundamental solution, $\varphi_\NNN(x,\omega)$,  defined as in Definition~\ref{def:fund}.
We have
\begin{eqnarray*}
\varphi_\NNN(0,\omega)&=\cos(\omega\tau)+\OOO{\frac{1}{\omega}},\qquad \textrm{as }\omega\rightarrow\infty,\\
(\varphi_\NNN)'(0,\omega)&=\omega\sin(\omega\tau)+\OOO{1},\qquad \textrm{as }\omega\rightarrow\infty,
\end{eqnarray*}
where $(\varphi_\NNN)'(0,\omega)$ denotes the spatial derivative $\frac{d}{dx}\varphi_{\NNN}(x,\omega)$ at $x=0$.
\end{lemma}
\begin{proof}
We start by writing $\varphi_\NNN(x,\omega)$ in its integral representation. Indeed, the fundamental solution $\varphi_{\NNN}(x,\omega)$ is given by~\cite{marchenko-book}:
\begin{equation*}
\varphi_{\NNN}(x,\omega)=\cos(\omega(\tau-x))+\int_x^{\tau} G_{\NNN}(\tau-x,\tau-s)\cos(\omega(\tau-s))ds,
\end{equation*}
where $G_{\NNN}(x,y)$ is a real function with the same regularity as $\int_x^\tau q(s) ds$: here, it is $C^1$ with respect to both coordinates.
We note that, as $G_\NNN$ is in $C^1$, by integrating by parts, one has:
\begin{eqnarray*}
\fl \int_0^\tau G_\NNN(\tau,\tau-s)\cos(\omega(\tau-s))ds=-\frac{1}{\omega}G_\NNN(\tau,\tau-s)\sin(\omega(\tau-s))\Big|_{s=0}^{s=\tau}\\ \qquad\qquad\qquad\qquad+\frac{1}{\omega}\int_0^\tau \frac{d}{ds}(G_\NNN(\tau,\tau-s))\sin(\omega(\tau-s))ds=\OOO{\frac{1}{\omega}}.
\end{eqnarray*}
Therefore $\varphi_\NNN(0,\omega)=\cos(\omega\tau)+\OOO{1/\omega}$ as $\omega\rightarrow\infty$ and we have the first relation.

For the spatial derivative $(\varphi_\NNN)'(x,\omega)$ at the point $x=0$, we have:
$$
(\varphi_\NNN)'(0,\omega)=\omega\sin(\omega(\tau))-G(\tau,\tau)\cos(\omega \tau)+\int_0^\tau \frac{d}{dx}(G_\NNN(\tau-x,\tau-s))\Big|_{x=0}\cos(\omega(\tau-s))ds.
$$
The kernel $G(x,y)$ being $C^1$, we have
$$
-G(\tau,\tau)\cos(\omega \tau)+\int_0^\tau \frac{d}{dx}(G_\NNN(\tau-x,\tau-s))\Big|_{x=0}\cos(\omega(\tau-s))ds=\OOO{1}\qquad \textrm{as }\omega\rightarrow\infty.
$$
Thus, $(\varphi_\NNN)'(0,\omega)= \omega\sin(\omega(\tau-x))+\OOO{1}$ as $\omega\rightarrow\infty$ and the second relation follows.
\end{proof}
As we see by Lemma~\ref{lem:fund}, in the high-frequency regime, the fundamental solutions become asymptotically independent of the potential $q$. The next lemma, provides an explicit method to identify the number $N$ and the lengthes $(\tau_j)_{j=1}^N$ of the branches for the homogenous case, where we know that the potentials $(q_j)_{j=1}^N$ are all zero.
\begin{lemma}\label{lem:homogenous}
Consider a star-shape network $\Gamma$ composed of $n_j$ branches of length $\tau_j$ ($j=1,\cdots,m$) all joining at a central node so that the whole number of branches $N$ is given by $\sum_{j=1}^m n_j$. Assume the potential $q$ on the network to be 0 ($q\equiv 0$). Then the knowledge of the Neumann reflection coefficient $R_\NNN(\omega)$ determines uniquely the parameters $(n_j)_{j=1}^m$ and $(\tau_j)_{j=1}^m$.
\end{lemma}
\begin{proof}
We need to apply the explicit computation of the reflection coefficient provided by~\eref{eq:R}. The fundamental solutions are given, simply, by $\varphi_\NNN^j(x,\omega)=\cos(\omega(\tau_j-x))$. Therefore:
$$
\frac{1-R_{\NNN}(\omega)}{1+R_{\NNN}(\omega)}=\frac{1}{\imath\omega}H-\frac{1}{\imath\omega}\sum_{j=1}^m  n_j \frac{\omega\sin(\tau_j\omega)}{\cos(\tau_j\omega)}.
$$
The knowledge of $R_\NNN(\omega)$ determines uniquely the signal:
$$
f(\omega):=\sum_{j=1}^m n_j\tan(\omega\tau_j).
$$
Assuming, without loss of generality, that the lengthes $\tau_j$ are ordered increasingly $\tau_1<\cdots<\tau_m$, the first pole of the function $f(\omega)$ coincides with $\pi/2\tau_m$ and therefore determines $\tau_m$. Furthermore,
$$
n_m=\lim_{\omega\rightarrow \pi/2\tau_m} \cos(\omega\tau_m)f(\omega),
$$
and therefore one can also determine $n_m$. Now, considering the new signal $g(\omega)=f(\omega)-n_m\tan(\omega\tau_m)$, one removes the branches of length $\tau_m$ and exactly in the same manner, one can determine $\tau_{m-1}$ and $n_{m-1}$. The proof of the lemma  follows then by a simple induction.
\end{proof}
We are now ready to announce the main theorem of this section:
\begin{theorem}\label{thm:geometry}
Consider a star-shape network $\Gamma$ composed of $n_j$ branches of length $\tau_j$ ($j=1,\cdots,m$) all joining at a central node so that the whole number of branches $N$ is given by $\sum_{j=1}^m n_j$. Assume the potential $q$ on the network to be, simply, $C^0$. Then the knowledge of the Neumann reflection coefficient $R_\NNN(\omega)$ determines uniquely the parameters $(n_j)_{j=1}^m$ and $(\tau_j)_{j=1}^m$.
\end{theorem}
\begin{proof}
Assume that, there exists two graph settings $(\tau_j, q_j)_{j=1}^N$ and $(\tilde \tau_j, \tilde q_j)_{j=1}^{\tilde N}$ (the lengthes $\tau_j$ are not necessarily different) giving rise to the same Neumann reflection coefficients: $R_\NNN(\omega)\equiv \tilde R_\NNN(\omega)$. By the explicit formula~\eref{eq:R}, we have
$$
\frac{1}{\omega}\sum_{j=1}^N \frac{(\varphi_{\NNN}^j)'(0,\omega)}{\varphi_{\NNN}^j(0,\omega)}\equiv \frac{1}{\omega}\sum_{k=1}^{\tilde N} \frac{(\tilde\varphi_{\NNN}^k)'(0,\omega)}{\tilde\varphi_{\NNN}^k(0,\omega)}.
$$
This is equivalent to: \small
\begin{eqnarray}\label{eq:bla}
\fl \prod_{j=1}^{\tilde N}\tilde\varphi_\NNN^j(0,\omega)
\left(\sum_{k=1}^N(\varphi_\NNN^k)'(0,\omega)\prod_{l\neq k} \varphi_\NNN^l(0,\omega)\right)-\nonumber\\
\qquad\qquad\qquad\prod_{j=1}^{N}\varphi_\NNN^j(0,\omega)
\left(\sum_{k=1}^{\tilde N}(\tilde\varphi_\NNN^k)'(0,\omega)\prod_{l\neq k} \tilde\varphi_\NNN^l(0,\omega)\right)=0.
\end{eqnarray}\normalsize
Defining the function: \small
\begin{eqnarray*}
\fl F(\omega):= \prod_{j=1}^{\tilde N}\cos(\omega\tilde\tau_k)
\left(\sum_{k=1}^N\sin(\omega\tau_k)\prod_{l\neq k} \cos(\omega\tau_l)\right)-
\prod_{j=1}^{N}\cos(\omega\tau_j)
\left(\sum_{k=1}^{\tilde N}\sin(\omega\tilde\tau_k)\prod_{l\neq k} \cos(\omega\tilde\tau_l)\right),
\end{eqnarray*}\normalsize
the asymptotic formulas of Lemma~\ref{lem:fund} and the~\eref{eq:bla} imply
$$
F(\omega)=\OOO{1/\omega}\qquad \textrm{as }\omega \rightarrow\infty.
$$
However, the function $F(\omega)$ is a trigonometric polynomial and almost periodic in the Bohr's sense~\cite{besicovitch-54}. The function $F^2(\omega)$ is, also, almost periodic and furthermore, we have
\begin{eqnarray*}
\fl M(F^2):=\lim_{\Omega\rightarrow\infty}\frac{1}{\Omega}\int_0^\Omega F^2(\omega)d\omega
=\lim_{\Omega\rightarrow\infty}\frac{1}{\Omega}\left(\int_0^1 F^2(\omega)d\omega+\int_1^\Omega F^2(\omega)d\omega\right)\\
\qquad\qquad\qquad\qquad\qquad\qquad\qquad\leq \lim_{\Omega\rightarrow\infty}\frac{1}{\Omega}\left(C_1+C_2\int_1^\Omega \frac{1}{\omega^2}d\omega\right)=0.
\end{eqnarray*}
This, trivially, implies that $F=0$ (one only needs to apply the Parseval's Theorem to the generalized fourier series of the function $F$). However, the relation $F(\omega)\equiv 0$ is equivalent to
$$
\sum_{j=1}^N \tan(\omega\tau_j)=\sum_{j=1}^{\tilde N} \tan(\omega\tilde\tau_j),
$$
and therefore, by Lemma~\ref{lem:homogenous}, the two settings are equivalent and the theorem follows.
\end{proof}

\section{Inverse Scattering for potential identification}\label{sec:potential}
A second inverse problem, related to the detection of soft faults in the network, can be formulated as the identification of the potentials on the branches. Here, we consider the case of homogenous perfect wires. We will show that the measurement of the two reflection coefficients $R_{\DDD}(\omega)$ and $R_{\NNN}(\omega)$, corresponding, respectively, to a short circuit and an open circuit experiment, is enough to identify uniquely the small changes in the potential remaining in a certain regularity class. Indeed, we will prove that the inverse problem of finding the $H^1(\Gamma)$ potentials in an $\epsilon$ $L^\infty(\Gamma)$-neighborhood of the zero potential (homogenous case), is well-posed for $\epsilon$ small enough. In this aim, we will need an additive assumption on the electrical lengths $(\tau_j)_{j=1}^N$, in order to remove symmetries leading to degeneracy problems. However, as it will be discussed later, it seems that this assumption can be relaxed.

In a first result (Theorem~\ref{thm:equal}), we prove that under some natural assumption on the electrical lengths, the knowledge of only one reflection coefficient (here $R_{\NNN}(\omega)$) is sufficient to identify uniquely the values:
$$
\bar q_j=\int_0^{\aaa_j} q_j(s)ds\qquad \forall j=1,\cdots,N.
$$
If anyone of these quantities appear to be different from zero, we know that there has been a change of parameters in the corresponding branch. By performing classical inverse scattering techniques over this branch~\cite{Kay-72,Jaulent-82} we can identify its soft faults. The question is therefore to identify the soft faults in the network which does not change the quantities $(\bar q_j)_{j=1}^N$. This is treated in the two Theorems~\ref{thm:Borg} and~\ref{thm:half}. For both these theorems, we will need some more restrictive assumptions on the electrical lengths $(\aaa_j)_{j=1}^N$. In Theorem~\ref{thm:half}, we will show that the knowledge of only one reflection coefficient (here $R_{\NNN}(\omega)$) is enough to identify the potential when we know that it has not changed on the first half of the branches. The Theorem~\ref{thm:Borg} provides a well-posedness result for the inversion of the potential over the whole graph but necessitates the knowledge of both reflection coefficients $R_{\DDD}(\omega)$ and $R_{\NNN}(\omega)$.

In the sequel, we note that the potential over the infinite branch $e_0$ is always given and is 0. Indeed, this homogenous line is added only for the reflectometry experiment. The following theorem provides a global inversion result concerning the quantities $\bar q_j$.
\begin{theorem}\label{thm:equal}
Consider a star-shaped graph $\Gamma$ and assume that
\begin{description}
  \item[$\textbf{B1}$] The electrical lengths $\{\tau_j\}_{j=1}^N$ are not entire multiples of each other,
  $$\frac{\tau_j}{\tau_i}\notin\NN\qquad\textrm{ for } i\neq j.$$
\end{description}
If there exist two potentials $q=\otimes_{j=1}^N q_j$ and $\tilde q=\otimes_{j=1}^N \tilde q_j$ in $H^1(\Gamma)$ giving rise to the same reflection coefficient, $R_{\NNN}(\omega)\equiv\tilde R_{\NNN}(\omega)$, one necessarily has:
$$
\int_0^{\tau_j} q_j(s) ds=\int_0^{\tau_j} \tilde q_j(s) ds\qquad j=1,\cdots,N.
$$
\end{theorem}
\begin{remark}\label{rem:integral}
Note that the result of the Theorem is also valid if we have $R_{\DDD}(\omega)\equiv \tilde R_{\DDD}(\omega)$.
\end{remark}
This theorem allows us to identify the situations where the soft fault in the network causes a change of the quantities $\bar q_j$. As explained above, the inverse problem can then be considered on separate branches and solved through classical inverse scattering techniques. In the two following theorems, we assume that the soft faults in the network leave the quantities $\bar q_j$ unchanged and as for the perfect situation we are dealing with homogenous lines we will assume that
\begin{description}
  \item[$\textbf{B2}$] $\bar q_j=\int_0^{\tau_j} q_j(s)ds=0$ for $j=1,\cdots,N$.
\end{description}
\begin{theorem}\label{thm:half}
Consider a star-shaped graph $\Gamma$ and assume that
\begin{description}
  \item[$\textbf{B1}'$] For any $j,j'\in\{1,\cdots,N\}$ such that $j\neq j'$, $\tau_j/\tau_{j'}$ is an algebraic irrational number.
\end{description}
Then, there exists $\epsilon >0$ small enough such that, if the potentials $q$ and $\tilde q$ belonging to $H^1(\Gamma)$ and satisfying $\textbf{B2}$, $\| q\|_{L^\infty(\Gamma)},\|\tilde{q}\|_{L^\infty(\Gamma)}<\epsilon$ and $q_j(s)=\tilde q_j(s)$ for $s\in [0,\tau_j/2]$, $j=1,\cdots,N$, give rise to the same Neumann reflection coefficients, $R_{\NNN}(\omega)\equiv\tilde R_{\NNN}(\omega)$, then $q\equiv \tilde q$.
\end{theorem}
\begin{remark}\label{rem:half}
Once again, the result remains valid if we replace the Neumann reflection coefficient by the Dirichlet one.
\end{remark}
\begin{remark}\label{rem:B1}
The assumption $\textbf{B1}'$ seems very restrictive. However, as it will be seen through the proof, the only thing we need is that for any of fractions $\tau_j/\tau_{j'}$, there exists at most a finite number of co-prime factors $(p,q)\in\NN\times \NN$,  such that the Diophantine approximation
$$
\left|\frac{\tau_j}{\tau_{j'}}-\frac{p}{q}\right|\leq \frac{1}{q^3},
$$
holds true. However, this is a classical result of the Borel-Cantelli Lemma that for almost all (with respect to Lebesgue measure) positive real $\alpha$'s this Diophantine approximation has finite number of solutions.
\end{remark}
\begin{theorem}\label{thm:Borg}
Consider a star-shaped graph $\Gamma$ satisfying $\textbf{B1}'$.
There exists $\epsilon >0$ small enough such that, if the potentials $q$ and $\tilde q$ belonging to $H^1(\Gamma)$ and satisfying $\textbf{B2}$ and $\|q\|_{L^\infty(\Gamma)},\|\tilde{q}\|_{L^\infty(\Gamma)}<\epsilon$, give rise to the same Neumann and Dirichlet reflection coefficients,
$$R_{\NNN}(\omega)\equiv\tilde R_{\NNN}(\omega) \qquad\textrm{and}\qquad R_{\DDD}(\omega)\equiv\tilde R_{\DDD}(\omega),$$
then $q\equiv \tilde q$.
\end{theorem}
\begin{remark}
The Theorem~\ref{thm:Borg} is a natural generalization to the case of a graph of the Theorem of two boundary
spectra on an interval~\cite{Borg-46} (see, for instance, Theorem 1.4.4, Page 24,~\cite{yurko-book}). In~\cite{Borg-46}, Borg proved that the knowledge of two spectral data corresponding to two boundary conditions, determine uniquely
the potential on an interval. Here the reflection coefficients $R_{\NNN}(\omega)$ and
$R_{\DDD}(\omega)$ play the role of this spectral data (see the Subsection~\ref{ssec:sturm-liouville}).
\end{remark}

\subsection{Inverse Sturm-Liouville problem}\label{ssec:sturm-liouville}
Throughout this subsection, we will consider a general star-shaped graph $\Gamma$ (of $N$ finite branches) and a potential $q=\otimes_{j=1}^N q_j$ belonging to $H^1(\Gamma)$. Furthermore we assume for the potential $q$ that, the norm $\|\q\|_{L^\infty(\Gamma)}$ is sufficiently small.

The main objective of this subsection is to show that the knowledge of the reflection coefficient $R_{\NNN}(\omega)$ for $\LLL^+_{\NNN}$ (resp. $R_{\DDD}(\omega)$ for $\LLL^+_{\DDD}$) is equivalent to the knowledge of different positive spectra of Sturm-Liouville operators defined on $\Gamma$ with Neumann (resp. Dirichlet) boundary conditions at terminal nodes and for various boundary conditions at the central node. In fact, defining the function
$$
h_{\NNN,\DDD}(\omega)=H+\frac{\imath \omega(R_{\NNN,\DDD}(\omega)-1)}{(1+R_{\NNN,\DDD}(\omega))},
$$
where $H$ is given by~\eref{eq:H}, we have the following result.

\begin{prop}\label{prop:sturm}
Fix $\omega\in\RR$ and define the Schr\"odinger operators $\LLL_{\NNN,\DDD}(\omega)$ on the compact graph $\Gamma$ as follows:
\begin{eqnarray*}
\LLL_{\NNN,\DDD}(\omega)&=\otimes_{j=1}^N (-\frac{d^2}{dx^2}+q_j(x)),\\
D(\LLL_{\NNN,\DDD}(\omega))&= \textrm{closure of } C_{\omega;\NNN,\DDD}^\infty(\Gamma) \textrm{ in }
H^2(\Gamma),
\end{eqnarray*}
where $C_{\omega;\NNN}^\infty(\Gamma)$ (resp. $C_{\omega;\DDD}^\infty(\Gamma)$) denotes the space of infinitely
differentiable functions $f=\otimes_{j=1}^N f_j$ defined on $\Gamma$ satisfying the boundary
conditions
\begin{eqnarray*}
f_j(0)&=f_{j'}(0)=:\bar f \qquad &j,j'=1,\cdots,N,\\
\sum_{j=1}^N f_j'(0)&= h_{\NNN,\DDD}(\omega) \bar f,&\\
f_j'(\tau_j)&=0\quad (f_j(\tau_j)=0 \textrm{ for } C_{\omega;\DDD}^\infty(\Gamma)),\qquad  ~~& j=1,\cdots,N.
\end{eqnarray*}
Then we are able to characterize the positive spectrum of $\LLL_{\NNN,\DDD}(\omega)$ as a level set of the function $h_{\NNN,\DDD}(\omega)$:
\begin{equation*}
\sigma^+(\LLL_{\NNN,\DDD}(\omega))=\left\{ \xi^2~|~\xi\in\RR,~ h_{\NNN,\DDD}(\xi)=h_{\NNN,\DDD}(\omega)\right\}.
\end{equation*}
\end{prop}
\begin{proof}
We prove the proposition for the case of Neumann boundary conditions. The Dirichlet case can be treated exactly in the same manner.
We start by proving the inclusion
$$
\sigma^+(\LLL_{\NNN}(\omega))\subseteq\left\{ \xi^2~|~\xi\in\RR,~ h_{\NNN}(\xi)=h_{\NNN}(\omega)\right\}.
$$
Let $\xi^2\in\sigma^+(\LLL_{\NNN}(\omega))$, then there exists $\Psi$ eigenfunction of the operator
$\LLL_{\NNN}(\omega)$ associated to $\xi^2$. In particular, it satisfies
$$
\sum_{j=1}^{N}\Psi_j'(0)=h_{\NNN}(\omega)\bar\Psi,
$$
where $\bar\Psi$ is the common value of $\Psi$ at the central node.

Now we extend $\Psi$ to the extended graph $\Gamma^+$, such that
$\Psi^+$ is a scattering solution for $\LLL^+_{\NNN}$ (see the Proposition~\ref{prop:scattering}).
In particular, the function $\Psi^+$ must satisfy, at the central node,
\begin{eqnarray*}
\Psi^+_j(0)=\Psi^+_{0}(0),\qquad\qquad\quad\qquad j=1,\cdots,N,\\
\sum_{j=1}^{N}(\Psi^{+}_j)'(0)-(\Psi^+_{0})'(0)=H\Psi^+_0(0).
\end{eqnarray*}
Noting that $\Psi$ is an eigenfunction of $(\LLL_\NNN(\omega),D(\LLL_\NNN(\omega))$, we have
\begin{equation}\label{eq:aux}
h_{\NNN}(\omega)\Psi^+_0(0)-(\Psi^+_{0})'(0)=\sum_{j=1}^{N}(\Psi^{+}_j)'(0)-(\Psi^+_{0})'(0)=H\Psi^+_0(0).
\end{equation}
Now, noting that $\Psi^+$ over the infinite branch admits the following form
$$
\Psi^+_{0}(x)=R_{\NNN}(\xi)e^{\imath\xi x}+e^{-\imath\xi x}\qquad x\in[0,\infty),
$$
the relation~\eref{eq:aux} yields to
$$
h_{\NNN}(\omega)(R_{\NNN}(\xi)+1)-\imath\xi(R_{\NNN}(\xi)-1)=H (R_\NNN(\xi)+1),
$$
or equivalently
$$
h_{\NNN}(\omega)=H+\frac{\imath\xi(R_{\NNN}(\xi)-1)}{(R_{\NNN}(\xi)+1)}=h_{\NNN}(\xi).
$$
This proves the first inclusion. Now, we prove that
$$
\sigma^+(\LLL_{\NNN}(\omega))\supseteq\left\{ \xi^2~|~\xi\in\RR,~ h_{\NNN}(\xi)=h_{\NNN}(\omega)\right\}.
$$
Let $\xi\in\RR$ be such that $h_{\NNN}(\xi)=h_{\NNN}(\omega)$. We consider a scattering solution $\Psi^+$ of the extended operator $\LLL^+_{\NNN}$ (defined by~\eref{eq:schN}) associated to the frequency $\xi^2$. We, then, prove that the restriction of $\Psi^+$ to the compact graph $\Gamma$ is an eigenfunction of $\LLL_{\NNN}(\omega)$ associated to the eigenvalue $\xi^2$. This trivially implies that $\xi^2\in\sigma^+(\LLL_{\NNN}(\omega))$.

In this aim, we only need to show that this restriction of $\Psi^+$ to $\Gamma$ is in the domain $D(\LLL_{\NNN}(\omega))$. Indeed, this is equivalent to proving that the boundary condition:
\begin{equation}\label{eq:aux2}
\sum_{j=1}^N (\Psi^+_j)'(0)=h(\omega)\Psi^+_0(0),
\end{equation}
is satisfied. As $\Psi^+$ is a scattering solution of $\LLL^+_{\NNN}$, it satisfies
$$
\sum_{j=1}^N (\Psi^+_j)'(0)=H\Psi_0^+(0)+(\Psi_0^+)'(0)=\left(H+\frac{(\Psi_0^+)'(0)}{\Psi^+_0(0)}\right)\Psi_0^+(0).
$$
Furthermore,
$$
\Psi_0^+(0)=R_{\NNN}(\xi)+1\qquad\textrm{and}\qquad (\Psi_0^+)'(0)=\imath\xi(R_{\NNN}(\xi)-1),
$$
and so
$$
\sum_{j=1}^N (\Psi^+_j)'(0)=\left(H+\frac{\imath\xi(R_{\NNN}(\xi)-1)}{R_{\NNN}(\xi)+1}\right)\Psi_0^+(0)=h(\xi)\Psi_0^+(0)=h(\omega)\Psi_0^+(0).
$$
This proves~\eref{eq:aux2} and finishes the proof of the proposition.
\end{proof}
We have shown that, the knowledge of the reflection coefficient $R_{\NNN,\DDD}(\omega)$ for $\Gamma^+$ is equivalent to the knowledge of
the positive part of the spectra for $\LLL_{\NNN,\DDD}$ for all values $h$ at the central node.

The following proposition provides the characteristic equation permitting to identify the eigenvalues of the operator $\LLL_{\NNN,\DDD}(\omega)$:
\begin{prop}\label{prop:char}
The real $\lambda^2>0$ is an eigenvalue of the operator $\LLL_{\NNN,\DDD}(\omega)$ if and only if
\begin{equation*}
\Psi_{\NNN,\DDD}(\lambda)=h_{\NNN,\DDD}(\omega)\Phi_{\NNN,\DDD}(\lambda),
\end{equation*}
where
\begin{equation}\label{eq:aux3}
\fl \Phi_{\NNN,\DDD}(\lambda):=\prod_{j=1}^N \varphi_{\NNN,\DDD}^j(0,\lambda)\quad \textrm{and}\quad \Psi_{\NNN,\DDD}(\lambda):=\frac{d}{dx}\left(\prod_{j=1}^N \varphi_{\NNN,\DDD}^j(x,\lambda)\right)\Big|_{x=0},
\end{equation}
$\varphi_{\NNN,\DDD}^j(x,\lambda)$ being the fundamental solutions on different branches.
\end{prop}
\begin{proof}
We give the proof for the Neumann boundary conditions, noting that the Dirichlet case can be treated, exactly, in the same manner.
Assume $\lambda^2$ to be a positive eigenvalue of $\LLL_{\NNN}(\omega)$. The associated eigenfunction, $y_\lambda(x)=\otimes_{j=1}^N y_\lambda^j(x)$, has necessarily the following form:
\begin{equation*}
y_\lambda^j(x)=\alpha_j \varphi_{\NNN}^j(x,\lambda),
\end{equation*}
where $\alpha_j$'s are real constants and the vector $(\alpha_1,\cdots,\alpha_N)$ is different from zero. The function $y_{\lambda}$, being in the domain $D(\LLL_{\NNN}(\omega))$, it should satisfy the associated boundary condition at the central node. This implies that the vector $(\alpha_1,\cdots,\alpha_N)$ is in the kernel of the matrix:
\small
$$
M:=\left(\begin{tabular}{@{}lllll}
$\varphi_{\NNN}^1 (0,\lambda)$ & $-\varphi_{\NNN}^2 (0,\lambda)$&~0&$\cdots$ & 0\\
 0&$~\varphi^2_\NNN (0,\lambda)$ & $-\varphi^3_\NNN (0,\lambda)$&$\cdots$&0\\
 0&~0&$~\varphi^3_\NNN (0,\lambda)$&$\cdots$&0  \\
 $\cdots$&$~\cdots$&$~\cdots$&$\cdots$&$\cdots$\\
 $-h_{\NNN}(\omega)\varphi^1_\NNN (0,\lambda)+\psi^1_\NNN(0,\lambda)$ &$~\psi^2_\NNN(0,\lambda)$&$~\psi^3_\NNN
 (0,\lambda)$&$\ldots$&
 $\psi^N_\NNN(0,\lambda)$
\end{tabular}\right)
$$\normalsize
where $\psi^j_\NNN(0,\lambda)$ denotes $\frac{d}{dx}\varphi^j_\NNN(x,\lambda)|_{x=0}$. This means that the determinant $\det(M)$ is necessarily 0. Developing this determinant we find:
$$
\Psi_\NNN(\lambda)=h_\NNN(\omega)\Phi_\NNN(\lambda).
$$
\end{proof}
\begin{corollary}\label{cor:char}
Consider two potentials $q=\otimes_{j=1}^N q_j$ and $\tilde q=\otimes_{j=1}^N q_j$ and denote by $\LLL^+_{\NNN}$ and  $\tilde\LLL^+_{\NNN}$, the associated Neumann  Schr\"odinger operators defined on the extended graph $\Gamma^+$. Assuming that the of the reflection coefficients $R_{\NNN}(\omega)$ and $\tilde R_{\NNN}(\omega)$ are equivalent $R_{\NNN}(\omega)\equiv\tilde R_{\NNN}(\omega)$, we have
\begin{equation}\label{eq:char2}
\Phi_\NNN(\omega)\tilde\Psi_\NNN(\omega)=\tilde \Phi_\NNN(\omega)\Psi_\NNN(\omega),\qquad \forall \omega\in\RR,
\end{equation}
where $\Phi_\NNN$, $\Psi_\NNN$, $\tilde\Phi_\NNN$ and $\tilde\Psi_\NNN$ are defined through~\eref{eq:aux3} for the potentials $q$ and $\tilde q$.
\end{corollary}
\begin{proof}
By Proposition~\ref{prop:sturm}, $\omega^2$ is an eigenvalue of the operator $\LLL_\NNN(\omega)$ and $\tilde\LLL_\NNN(\omega)$. Applying the Proposition~\ref{prop:char}, this means that
$$
\Psi(\omega)=h_\NNN(\omega)\Phi_\NNN(\omega)\qquad\textrm{and}\qquad \tilde\Psi(\omega)=\tilde h_\NNN(\omega)\tilde\Phi_\NNN(\omega).
$$
As $R_\NNN(\omega)\equiv \tilde R_\NNN(\omega)$, we have $h_\NNN(\omega)\equiv \tilde h_\NNN(\omega)$ and thus the above equation yields to~\eref{eq:char2}.
\end{proof}
The above corollary is also valid when we replace the Neumann by Dirichlet boundary conditions. Finally, this corollary yields to the following proposition on the difference between the two potentials $q$ and $\tilde q$.
\begin{prop}\label{prop:char2}
Consider two potentials $q=\otimes_{j=1}^N q_j$ and $\tilde q=\otimes_{j=1}^N q_j$ and denote by $\LLL^+_{\NNN}$ and  $\tilde\LLL^+_{\NNN}$, the associated Neumann Schr\"odinger operators defined on the extended graph $\Gamma^+$. Assuming that the of the reflection coefficients $R_{\NNN}(\omega)$ and $\tilde R_{\NNN}(\omega)$ are equivalent $R_{\NNN}(\omega)\equiv\tilde R_{\NNN}(\omega)$, we have
\begin{equation}\label{eq:charfinal}
\fl\qquad\qquad\sum_{j=1}^N \prod_{k\neq j} \varphi_\NNN^k(0,\omega)\tilde \varphi_\NNN^k(0,\omega)\int_0^{\tau_j} \hat q_j(x) \varphi_\NNN^j(x,\omega)\tilde \varphi_\NNN^j(x,\omega)dx =0,\qquad \forall \omega\in\RR,
\end{equation}
where $\hat q_j=\tilde q_j-q_j$.
\end{prop}
\begin{proof}
For $j=1,\cdots,N$, we have:
\begin{eqnarray*}
\int_0^{\aaa_j} \tilde q(x)&\tilde \varphi^j_\NNN(x,\omega)\varphi_\NNN^j(x,\omega)dx-\int_0^{\aaa_j} q_j(x)
\varphi_\NNN^j(x,\omega)\tilde \varphi_\NNN^j(x,\omega)dx=\\
&=\varphi_\NNN^j(x,\omega) \frac{d}{dx}\tilde{\varphi}_\NNN^j(x,\omega)\big|_{x=0}^{x=\tau_j}- \frac{d}{dx}\varphi_\NNN^j(x,\omega)\tilde{\varphi}_\NNN^j(x,\omega)\big|_{x=0}^{x=\tau_j}\\
&=\psi_\NNN^j(0,\omega)\tilde{\varphi}_\NNN^j(0,\omega)-\varphi_\NNN^j(0,\omega)\tilde\psi_\NNN^j(0,\omega).
\end{eqnarray*}
Here the second line has been obtained from the first one, replacing $q_j(x)\varphi^j_\NNN(x,\omega)$ by $
\frac{d^2}{dx^2}\varphi^j_\NNN(x,\omega)+\omega^2\varphi^j_\NNN(x,\omega)$ and integrating by parts. Using~\eref{eq:char2} and the above equation, we have:
\begin{eqnarray*}
&\sum_{j=1}^N \prod_{k\neq j} \varphi_\NNN^k(0,\omega)\tilde \varphi_\NNN^k(0,\omega)\int_0^{\tau_j} \hat q_j(x) \varphi_\NNN^j(x,\omega)\tilde \varphi_\NNN^j(x,\omega)dx=\\
&\sum_{j=1}^N \prod_{k\neq j} \varphi_\NNN^k(0,\omega)\tilde \varphi_\NNN^k(0,\omega)\left(\psi_\NNN^j(0,\omega) \tilde{\varphi}_\NNN^j(0,\omega)-\varphi_\NNN^j(0,\omega)\tilde\psi_\NNN^j(0,\omega)\right)=\\
&\Psi_\NNN(\omega)\tilde\Phi_\NNN(\omega)-\Phi_\NNN(\omega)\tilde\Psi_\NNN(\omega)=0.
\end{eqnarray*}
\end{proof}
Before finishing this subsection, note that, once more, the above proposition is also valid for the case of Dirichlet boundary conditions and $R_\DDD(\omega)\equiv \tilde R_\DDD(\omega)$ implies:
\begin{equation}\label{eq:charfinalD}
\fl\qquad\qquad\sum_{j=1}^N \prod_{k\neq j} \varphi_\DDD^k(0,\omega)\tilde \varphi_\DDD^k(0,\omega)\int_0^{\tau_j} \hat q_j(x) \varphi_\DDD^j(x,\omega)\tilde \varphi_\DDD^j(x,\omega)dx =0,\qquad \forall \omega\in\RR.
\end{equation}
We are now ready to prove the Theorems~\ref{thm:equal},~\ref{thm:half} and~\ref{thm:Borg}.

\subsection{Proof of Theorem \ref{thm:equal}}\label{ssec:equal}
We prove the Theorem~\ref{thm:equal} applying the characteristic equation~\eref{eq:charfinal} and the integral representation of the functions $\varphi^j_{\NNN,\DDD}(x,\omega)$. Again, for simplicity sakes, we give the proof only for the case of Neumann boundary conditions, noting that the Dirichlet case can be done in a similar way.

Similarly to the Section~\ref{sec:geometry}, the fundamental solutions $\varphi^j_{\NNN}(x,\omega)$ are given by~\cite{marchenko-book}:
\begin{equation*}
\varphi^j_{\NNN}(x,\omega)=\cos(\omega(\tau_j-x))+\int_x^{\tau_j} G_{\NNN}^j(\tau_j-x,\tau_j-s)\cos(\omega(\tau_j-s))ds,
\end{equation*}
where $G_{\NNN}^j(x,y)$ are $C^1$ with respect to  the both coordinates.

The above equation yields to the following representation for the product functions:
\begin{equation}\label{eq:Volterra}
\fl \varphi^j_{\NNN}(x,\omega)\tilde{\varphi}^j_\NNN(x,\omega)=\cos^2(\omega(\tau_j-x))+\frac{1}{2}\int_{x}^{\tau_j}
 K^j_\NNN(\tau_j-x,\tau_j-s)\cos(2\omega(\tau_j-s)) ds
\end{equation}
where $K^j_\NNN(x,y)$ is a Volterra kernel, i.e. $K_j(x,y)\equiv 0$ if $y>x$ and
\begin{eqnarray*}
\fl K_\NNN^j(x,y)=2[G_\NNN^j(x,x-2y)+\tilde{G}_\NNN^j(x,x-2y)]+\\
+\int_{2y-x}^x\tilde{G}_\NNN^j(x,s)G_\NNN^j(x,s-2y)ds+\int_{-x}^{x-2y}\tilde{G}_\NNN^j(x,s)G_\NNN^j(x,s+2y)ds.
\end{eqnarray*}
At this point, we note that, as $K_\NNN^j(x,y)$  is a $C^1$ function, we have
\begin{equation}\label{eq:lim}
\fl\qquad\qquad
\frac{1}{2}\int_{x}^{\tau_j}
 K^j_\NNN(\tau_j-x,\tau_j-s)\cos(2\omega(\tau_j-s)) ds=\OOO{\frac{1}{\omega}},\qquad x\in[0,\tau_j].
\end{equation}
Applying the characteristic equation~\eref{eq:charfinal} and developing the products $\varphi^j_{\NNN}(x,\omega)$ $\tilde{\varphi}^j_\NNN(x,\omega)$ by~\eref{eq:Volterra}, and finally using~\eref{eq:lim}, we have:
\begin{eqnarray}
\sum_{j=1}^{N}\left(\prod_{k\neq j }\cos^2(\omega\aaa_k)\right)
\int_{0}^{\aaa_j} \hat{\q}_j(s)\cos^2(\omega(\tau_j-s)) ds=\OOO{\frac{1}{\omega}},\nonumber\\
\sum_{j=1}^{N}\left(\prod_{k\neq j }\cos^2(\omega\aaa_k)\right)
\int_{0}^{\aaa_j} \hat{\q}_j(s)\left(\frac{1+\cos2(\omega(\tau_j-s))}{2}\right)ds=\OOO{\frac{1}{\omega}},\nonumber\\
\sum_{j=1}^{N}\left(\prod_{k\neq j }\cos^2(\omega\aaa_k)\right)
\frac{1}{2}\int_{0}^{\aaa_j}\hat{\q}_j(s)ds=\OOO{\frac{1}{\omega}}.\label{eq:quasiperiodic}
\end{eqnarray}
In the last passage, we applied the fact that $\int_0^{\tau_j}\hat q_j(s)\cos 2(\omega(\tau_j-s))ds=\OOO{1/\omega}$, since
$\hat q$ is in $H^1(\Gamma)$.\\
The left side of~\eref{eq:quasiperiodic} is an almost periodic function with respect to $\omega$, in the Bohr's sense.
Following the same arguments as those of the Theorem \ref{thm:geometry} we obtain
$$
\sum_{j=1}^{N}\left(\prod_{k\neq j }\cos^2(\omega\aaa_k)\right)
\frac{1}{2}\int_{0}^{\aaa_j}\hat{\q}_j(s)ds=0.
$$
 Choosing $\omega_j=\pi/2\tau_j$ and noting that, as the parameters $(\tau_j)_{j=1}^N$ are not entire multiples of each other i.e. $\tau_i/\tau_j\notin\NN$:
$$cos(\omega_j\tau_i)=cos(\frac{\pi}{2}\frac{\tau_i}{\tau_j})\neq 0\qquad\textrm{for } i\neq j.$$
For each value $\omega_j$, we have
\begin{equation*}
\prod_{k\neq j}\cos^2(\omega_j\tau_k)\int_0^{\tau_j} \hat q_j(s) ds =0\qquad \Rightarrow\qquad \int_0^{\tau_j} \hat q_j(s) ds =0.
\end{equation*}
and finishes the proof of the Theorem~\ref{thm:equal}.

\subsection{Proof of Theorem \ref{thm:half}}\label{ssec:half}
In this subsection, we consider two potentials $q=\otimes_{j=1}^N q_j$ and  $\tilde q=\otimes_{j=1}^N \tilde q_j$, satisfying the assumptions of the Theorem~\ref{thm:half}. Assuming that they give rise to the same Neumann reflection coefficients, $R_\NNN(\omega)\equiv \tilde R_\NNN(\omega)$, we have the characteristic equation~\eref{eq:charfinal}.

Let us define the operator $\LLL^j_{\NNN}$ to be the operator $-\frac{d^2}{dx^2}+q_j(x)$ over $[0,\tau_j]$ with the domain
$$
D(\LLL^j_{\NNN})=\textrm{closure of } C^\infty_\NNN(0,\tau_j) \textrm{ in } H^2(0,\tau_j),
$$
where $C^\infty_\NNN(0,\tau_j)$ denotes the space of infinitely differentiable functions $f$ defined on $[0,\tau_j]$ satisfying Dirichlet boundary condition at 0 and Neumann boundary condition at $\tau_j$. By the perturbation theory for linear operators~\cite{kato-book-80}, there exists $\epsilon_1>0$ small enough such that, if $\|q\|_{L^\infty(\Gamma)},\|\tilde q\|_{L^\infty(\Gamma)}<\epsilon_1$ then the eigenvalues of $\LLL^j_{\NNN}$ remain positive.

Considering $((\lambda_i^j)^2)_{i=1}^\infty$ ($\lambda_i^j>0$) the sequence of eigenvalues of $\LLL^j_{\NNN}$,~\eref{eq:charfinal} implies
for each $ j=1,\cdots,N,$
\begin{equation}\label{eq:aux5}
\fl\qquad\quad
\prod_{k\neq j} \varphi_\NNN^k(0,\lambda_i^j)\tilde \varphi_\NNN^k(0,\lambda_i^j)\int_0^{\tau_j} \hat q_j(x) \varphi_\NNN^j(x,\lambda_i^j)\tilde \varphi_\NNN^j(x,\lambda_i^j)dx =0,\quad \forall i=1,2,\cdots
\end{equation}
where we have applied the fact that $\varphi_\NNN^j(0,\lambda_i^j)=0$.

At this point, we will use the assumption~$\textbf{B1}'$ on the lengthes $\tau_j$ to obtain a Lemma on the non-overlapping of the eigenvalues for different branches:
\begin{lemma}\label{lem:separation}
Under the assumptions of the Theorem~\ref{thm:half}, there exists $\epsilon_2>0$ small enough such that, if $\|q\|_{L^\infty(\Gamma)},\|\tilde q\|_{L^\infty(\Gamma)}<\epsilon_2$, then
$$
\prod_{k\neq j} \varphi_\NNN^k(0,\lambda_i^j)\tilde \varphi_\NNN^k(0,\lambda_i^j)\neq 0,\qquad \forall j=1,\cdots,N,\qquad\forall i=1,2,\cdots.
$$
\end{lemma}
\begin{proof}
In order to prove this Lemma, we only need to show that  $\lambda_i^j$ is not an eigenvalue of $\LLL_\NNN^k$ nor $\tilde\LLL_\NNN^k$ for $k\neq j$.

In this aim, we first show that there exists $M>0$ such that for integers $i_1,i_2>M$, $\lambda_{i_1}^j$ is different from $\lambda_{i_2}^k$ and $\tilde \lambda_{i_2}^k$ the eigenvalues of $\LLL_\NNN^k$ and $\tilde\LLL_\NNN^k$. Assume, contrarily, that, for all $M>0$, there exists $i_1,i_2>M$, such that
\begin{equation}\label{eq:aux4}
\lambda_{i_1}^j=\lambda_{i_2}^k\qquad\textrm{or}\qquad \lambda_{i_1}^j=\tilde \lambda_{i_2}^k.
\end{equation}
As the potentials $q$ and $\tilde q$ are $H^1$, we have the following asymptotic formula's for the eigenvalues (see, for instance, Remark 1.1.1, page 7,~\cite{yurko-book}):
\begin{eqnarray*}
\lambda_i^k &= \frac{(2i-1)\pi}{2\tau_k}+\frac{\int_0^{\tau_k}q_k(s)ds}{2i\tau_k}+\OOO{\frac{1}{i^2}},\\
\tilde \lambda_i^k &= \frac{(2i-1)\pi}{2\tau_k}+\frac{\int_0^{\tau_k}\tilde q_k(s)ds}{2i\tau_k}+\OOO{\frac{1}{i^2}}.\\
\end{eqnarray*}
However, as by assumption $\textbf{B2}$, the integrals $\int_0^{\tau_k}q_k(s)ds$ and $\int_0^{\tau_k}\tilde q_k(s)ds$ are zero, this implies:
\begin{equation}\label{eq:asy}
\lambda_i^k = \frac{(2i-1)\pi}{2\tau_k}+\OOO{\frac{1}{i^3}}\quad\textrm{and} \quad \tilde \lambda_i^k = \frac{(2i-1)\pi}{2\tau_k}+\OOO{\frac{1}{i^2}}.
\end{equation}
Therefore, the relation~\eref{eq:aux4} implies that, for all $M>0$ there exists $i_1,i_2>M$ such that
$$
\left|\frac{(2i_1-1)\pi}{2\tau_j}-\frac{(2i_2-1)\pi}{2\tau_k}\right|\leq \OOO{\frac{1}{i_1^2}}+\OOO{\frac{1}{i_2^2}}.
$$
Assuming, without loss of generality, that $i_1\leq i_2$ and dividing the above inequality by $(2i_1-1)\pi/2\tau_k$, we have
$$
\left|\frac{\tau_k}{\tau_j}-\frac{2i_2-1}{2i_1-1}\right|\leq \OOO{\frac{1}{(2i_1-1)^3}}.
$$
Therefore, we must have the existence of an infinite number of integer couples $(i_1,i_2)$ satisfying the above inequality. However, by Thue-Siegel-Roth Theorem~\cite{roth-55}, for any irrational algebraic number $\alpha$, and for any $\delta>0$, the inequality
$$
\left|\alpha-p/q\right|<1/|q|^{2+\delta},
$$
has only a finite number of integer solutions $p,q$ ($q\neq 0$). This, trivially, leads to a contradiction and therefore there exists $M>0$ such that for $i_1,i_2>M$ and $j\neq k$, $\lambda^j_{i_1}\neq\lambda^k_{i_2}$ and $\lambda^j_{i_1}\neq\tilde\lambda^k_{i_2}$.

For the $M$ first eigenvalues on each branch, we apply the perturbation theory for linear operators~\cite{kato-book-80}. Having $j$, the branch index, fixed, we will show that for $\epsilon^j_2>0$ small enough, if $\|q\|_{L_\infty},\|\tilde q\|_{L_\infty}<\epsilon^j_2$ then the quantities $(\lambda_i^j)_{i=1}^M$  do note coincide with the quantities $(\lambda_i^k)_{i=1}^\infty$, where $k\neq j$. In fact, for the case of $q=\tilde q=0$, this is, trivially, a  consequence of the fact that the branch lengthes $\tau_j$ and $\tau_k$ are 2-by-2 $\QQ$-linearly independent. Now, adding a small perturbation, $q$ or $\tilde q$, in the generalized sense (see page 206~\cite{kato-book-80}), by the continuity of a finite system of eigenvalues (see page 213~\cite{kato-book-80}), this claim remains valid. Thus, there exists a small enough $\epsilon_2=\min_{j=1,\cdots,N}(\epsilon_2^j)$ such that, if $\|q\|_{L_\infty},\|\tilde q\|_{L_\infty}<\epsilon_2$ then $(\lambda_i^j)_{i=1}^M$ do not coincide with  $(\lambda_i^k)_{k=1}^\infty$ , for all $j,k=1,\cdots,N$ verifying $j\neq k$.
\end{proof}
Applying Lemma~\ref{lem:separation} to the \Eref{eq:aux5}, for $\|q\|_{L^\infty},\|\tilde q\|_{L^\infty}<\epsilon:=\min(\epsilon_1,\epsilon_2)$, we have:
\begin{equation*}
\fl\qquad\int_0^{\tau_j} \hat q_j(x) \varphi^j_{\NNN}(x,\lambda_i^j)\tilde \varphi^j_{\NNN}(x,\lambda_i^j) dx=0, \qquad \forall j=1,\cdots,N,~~~ \forall i=1,2,\cdots.
\end{equation*}
Using the integral representation with a Volterra kernel \eref{eq:Volterra}, we rewrite\small
\begin{eqnarray*}
\fl\int_0^{\aaa_j}\hat{\q}_j(x)\varphi^j_\NNN(x,\lambda^j_i) \tilde{\varphi}_\NNN^j(x,\lambda^j_i)dx =\\
\fl\qquad\qquad\int_0^{\aaa_j} \hat{\q}_j(x)\left(\cos^2 \lambda^j_i(\tau_j-x)+\int_x^{\aaa_j}\frac{K_\NNN^j(\tau_j-x,\tau_j-t)}{2}\cos 2\lambda^j_i(\tau_j-t) d t
\right) dx=\\
\fl\qquad\qquad\int_0^{\aaa_j} \hat{\q}_j(x)\left(\cos 2\lambda^j_i (\tau_j-x)+\int_x^{\aaa_j}\frac{K_\NNN^j(\tau_j-x,\tau_j-t)}{2}\cos 2\lambda^j_i (\tau_j-t) d t
\right) dx=\\
\fl\qquad\qquad\int_0^{\aaa_j} \left(\hat{\q}_j(t)+\int_0^{t}\frac{K_\NNN^j(\tau_j-t,\tau_j-x)}{2}\hat{\q}_j(x)dx \right)
 \cos 2\lambda^j_i(\tau_j-t )d t
\end{eqnarray*}\normalsize
where we have applied the assumption $\textbf{B2}$, $\int_0^{\tau_j} \hat q_j=0$, for the passage from the second to the third line and the Fubini Theorem for the last passage. This implies that
$$
\int_0^{\aaa_j} A^j(t) \cos 2\lambda^j_i(\tau_j-t) dt=0\qquad \forall i=1,2,\cdots,
$$
where
\begin{equation*}
A^j(t):=\hat{\q}_j( t)+\int_0^{t}\frac{K^j_\NNN(\tau_j-t,\tau_j-x)}{2}\hat{\q}_j(x)dx.
\end{equation*}
By the assumption of Theorem \ref{thm:half}, the potentials $q_j(t)$ are zero on $[0,\tau_j/2]$ and consequently
$A^j(t)= 0$ on $[0,\tau_j/2]$. Therefore, we have
$$
\int_{\aaa_j/2}^{\tau_j} A^j(t) \cos 2\lambda^j_i(\tau_j-t) dt=0\qquad \forall i=1,2,\cdots.
$$
However, by the Proposition 1.8.6 of~\cite{yurko-book}, the system $\{\cos (2\lambda^j_i (\tau_j-.))\}_{i=1,2,\cdots}$ provides a Riesz basis over $L^2[\tau_j/2,\tau_j]$.
Consequently,
$$
\hat{\q}_j( t)+\int_0^{t}\frac{K^j_\NNN(\tau_j-t,\tau_j-x)}{2}\hat{\q}_j(x)dx=0\qquad t\in[0,\tau_j].
$$
This homogenous Volterra integral equation has only  the trivial solution $\hat q_j\equiv 0$ on $[0,\tau_j]$.
This, trivially, implies that for $\epsilon=\min(\epsilon_1,\epsilon_2)$, if $\|q\|_{L^\infty(\Gamma)},\|\tilde q\|_{L^\infty(\Gamma)}<\epsilon$ then $\hat q \equiv 0$ on $\Gamma$.


\subsection{Proof of Theorem \ref{thm:Borg}}\label{ssec:Borg}

We consider two potentials $q=\otimes_{j=1}^N q_j$ and  $\tilde q=\otimes_{j=1}^N \tilde q_j$, satisfying the assumptions of the Theorem~\ref{thm:Borg}. Assuming that they give rise to the same Neumann and Dirichlet reflection coefficients, $R_\NNN(\omega)\equiv \tilde R_\NNN(\omega)$ and $R_\DDD(\omega)\equiv \tilde R_\DDD(\omega)$, we have the characteristic equations~\eref{eq:charfinal} and~\eref{eq:charfinalD}.

We define the operator $\LLL^j_{\DDD}$ exactly as $\LLL^j_{\NNN}$ (defined in the previous subsection) with Dirichlet boundary conditions at 0 and at $\tau_j$. Still, by the perturbation theory for linear operators~\cite{kato-book-80}, there exists $\epsilon_1>0$ small enough such that, if $\|q\|_{L^\infty(\Gamma)},\|\tilde q\|_{L^\infty(\Gamma)}<\epsilon_1$ then the eigenvalues of $\LLL^j_{\NNN}$ and $\LLL^j_{\DDD}$ are all positive.

Considering $((\lambda_i^j)^2)_{i=1}^\infty$ ($\lambda_i^j>0$) the sequence of eigenvalues of $\LLL^j_{\NNN}$, and $((\mu_i^j)^2)_{i=1}^\infty$ ($\mu_i^j>0$) the sequence of eigenvalues of $\LLL^j_{\DDD}$, the equations \eref{eq:charfinal} and \eref{eq:charfinalD} imply:\small
\begin{equation*}
\eqalign{
\fl\prod_{k\neq j} \varphi_\NNN^k(0,\lambda_i^j)\tilde \varphi_\NNN^k(0,\lambda_i^j)\int_0^{\tau_j} \hat q_j(x) \varphi_\NNN^j(x,\lambda_i^j)\tilde \varphi_\NNN^j(x,\lambda_i^j)dx &=0, \nonumber\\
\fl\prod_{k\neq j} \varphi_\DDD^k(0,\mu_i^j)\tilde \varphi_\DDD^k(0,\mu_i^j)\int_0^{\tau_j} \hat q_j(x) \varphi_\DDD^j(x,\mu_i^j)\tilde \varphi_\DDD^j(x,\mu_i^j)dx &=0,}\qquad \forall j=1,\cdots,N, \forall i=1,2,\cdots,
\end{equation*}\normalsize
where we have applied the fact that $\varphi_\NNN^j(0,\lambda_i^j)=\varphi_\DDD^j(0,\mu_i^j)=0$. Following the same arguments as those of the Lemma~\ref{lem:separation}, there exists $\epsilon_2>0$ small enough, such that
$$
\eqalign{\prod_{k\neq j} \varphi_\NNN^k(0,\lambda_i^j)\tilde \varphi_\NNN^k(0,\lambda_i^j)\neq 0,\cr
\prod_{k\neq j} \varphi_\DDD^k(0,\mu_i^j)\tilde \varphi_\DDD^k(0,\mu_i^j)\neq 0,}\qquad j=1,\cdots,N,~~i=1,2,\cdots.
$$
Therefore, for each $j=1,\cdots,N$ we have ,
\begin{equation}\label{eq:aux7}
\fl<\hat q_j(\cdot), \varphi_\NNN^j(\cdot,\lambda_i^j)>_{L^2(0,\tau_j)}=<\hat q_j(\cdot), \varphi_\DDD^j(\cdot,\mu_i^j)>_{L^2(0,\tau_j)}=0,\forall i=1,2,\cdots.
\end{equation}
At this point we define the sequence $(u_i^j)_{j=1,\cdots,N,~i=0,1,2,\cdots}$ as follows
$$
\eqalign{
u_0^j({\textbf x})&=(0,\cdots,0,\underbrace{1}_{j\textrm{-th position}},0,\cdots,0), \qquad{\textbf x}\in \Gamma\cr
u_{2i-1}^j({\textbf x})&=(0,\cdots,0,\underbrace{\varphi^j_{\NNN}(x,\lambda_i^j)\tilde \varphi^j_{\NNN}(x,\lambda_i^j)-\frac{1}{2}}_{j\textrm{-th position}},0,\cdots,0), \qquad{\textbf x}\in \Gamma,~i\geq 1 \cr
u_{2i}^j({\textbf x})&=(0,\cdots,0,\underbrace{\varphi^j_{\DDD}(x,\mu_i^j)\tilde \varphi^j_{\DDD}(x,\mu_i^j)-\frac{1}{2}}_{j\textrm{-th position}},0,\cdots,0),\qquad{\textbf x}\in \Gamma,~i\geq 1}
$$
As $\int_0^{\tau_j} \hat q_j(s) ds=0$ by Assumption $\textbf{B2}$,~\eref{eq:aux7} implies
$$
<\hat q(\cdot),u_i^j(\cdot)>_{L^2(\Gamma)}=0,\qquad \forall j=1,\cdots,N,\forall i=0,1,2,\cdots.
$$
where
$$
<v(\cdot),w(\cdot)>_{L^2(\Gamma)}=\sum_{j=1}^N <v^j(\cdot),w^j(\cdot)>_{L^2(0,\tau_j)}.
$$
However, by Lemma~\ref{lem:riesz2} proved in the Appendix, the sequence $(u_i^j)_{j=1,\cdots,N,~i=0,1,2,\cdots}$ forms a Riesz basis over $L^2(\Gamma)$ and therefore $\hat q=0$, which finishes the proof of Theorem~\ref{thm:Borg}.

\ack
This work was supported by grants from the French National Research Agency (ANR project 0-DEFECT) and from DIGITEO (project DIAGS).

\appendix
\section{Riesz basis properties}
The goal of this appendix is to provide a Lemma~\ref{lem:riesz2} on the Riesz basis property for the sequences $(u_i^j)_{j=1,\cdots,N,~i=0,1,2,\cdots}$ defined in Subsection~\ref{ssec:Borg}. Note that this is a direct consequence of a result already proved in~\cite{Borg-46}. However, we provide a proof for the sake of completeness.

We first provide a classical result on Riesz sequences~\cite[page 91]{yurko-book}:
\begin{lemma}\label{lem:bar}
Let $\{f_j \}_{j=1}^\infty$ be a Riesz basis for the Banach space $B$. Let $\{g_j \}_{j=1}^\infty$ be quadratically close to $\{f_j \}_{j=1}^\infty$, i.e.
$$
\sum_{j=1}^\infty \|g_j-f_j\|^2<\infty.
$$
Then if the sequence $\{g_j\}_{j=1}^\infty$ is $\omega$-linearly independent or if it is complete in $B$, then it is, also, a Riesz basis for $B$.
\end{lemma}
Throughout the Appendix, we consider the sequence $(u_i^j)$ to be defined as in Subsection~\ref{ssec:Borg}.
We therefore have the following lemma:
\begin{lemma}\label{lem:riesz2}
Assuming $q$ and $\tilde q$ in $H^1(\Gamma)$, the sequence $(u_i^j)_{j=1,\cdots,N,~i=0,1,2,\cdots}$ provides a Riesz basis over the Hilbert space $L^2(\Gamma)$.
\end{lemma}
\begin{proof}[Proof of Lemma~\ref{lem:riesz2}]
Fixing the branch index $j$, we prove that $\{u^j_i(x)\}_{i=0}^\infty$ is a Riesz basis over $L^2(0,\aaa_j)$. Here, for simplicity sakes, we have identified the vector $u^j_i(x)$ with its $j$-th component and therefore $u^j_i(x)$ is not a vector anymore but rather a function in $L^2(0,\tau_j)$. We also remove the branch index $j$.

First, we prove that
$\{v_i(x):=\frac{\tau}{i\pi}\frac{d}{dx}u_i(x)\}_{i=1}^{\infty}$ is a Riesz basis over $L^2(0,\aaa)$. Note that
\begin{eqnarray*}
\fl\quad v_i(0)=\frac{\tau}{i\pi}\frac{d}{dx}(\varphi_{\NNN,\DDD})(x,\lambda_i)\tilde \varphi_{\NNN,\DDD}
(x,\lambda_i)\Big|_{x=0}
+\frac{\tau}{i\pi}\frac{d}{dx}(\tilde\varphi_{\NNN,\DDD})(x,\lambda_i)\varphi_{\NNN,\DDD}
(x,\lambda_i)\Big|_{x=0}=0\\
\fl\quad v_i(\tau_j)=\frac{\tau}{i\pi}\frac{d}{dx}(\varphi_{\NNN,\DDD})(x,\lambda_i)\tilde \varphi_{\NNN,\DDD}
(x,\lambda_i)\Big|_{x=\tau}
+\frac{\tau}{i\pi}\frac{d}{dx}(\tilde\varphi_{\NNN,\DDD})(x,\lambda_i)\varphi_{\NNN,\DDD}
(x,\lambda_i)\Big|_{x=\tau}=0.
\end{eqnarray*}
In fact, one can easily see that the functions $v_i(x)$ are solutions of
\begin{equation}\label{eqn:w}
\fl\qquad\qquad-v_i''+2(q(x)+\tilde q(x))v_i +\int_0^x N(x,s) v_i(s)ds=4\lambda_i^2 v_i,\qquad x\in(0,\tau),
\end{equation}
where
$$
N(x,s)=q'(x)+\tilde q'(x)+(q(x)-\tilde q(x))\int_s^x(\tilde q(\xi)-q (\xi))d\xi.
$$
In particular, we obtain that the functions $\{v_i(x)\}_{i=1}^\infty$ are the eigenfunctions of \eref{eqn:w} with boundary conditions $v_i(0)=v_i(\tau)=0$.

Then there exists a bi-orthonormal sequence $\{v^{*}_i(x)\}_{i=1}^\infty$ in $L^2[0,\aaa]$, eigenfunctions
of the adjoint operator
$$
-(v_i^*)''+2(q(x)+\tilde q(x))v_i^* +\int_x^{\tau} N(s,x) v_i^*(s)ds=4\lambda_i^2 v_i^*,\quad v^*(0)=v^*(\tau)=0.
$$
Thus $\{v_i(x)\}_{i=1}^\infty$ are $\omega$-linearly independent.\\
From the integral representation of fundamental solutions and from Lemma~\ref{lem:fund},
we have
$$v_i(x)=\sin(\frac{i\pi x}{\tau})+\OOO{\frac{1}{i}}\qquad\textrm{as }i\rightarrow\infty.$$
So
$$
\sum_{i=1}^\infty\|v_i(x)-\sin(\frac{i\pi x}{\tau})\|_{L^2(0,\tau)}^2<\infty
$$
i.e. $\{v_i(x)\}_{i=1}^\infty$ is quadratically close to $\{\sin(\frac{i\pi x}{\tau})\}_{i=1}^\infty$. By the virtue of the Lemma \ref{lem:bar}, this gives that $\{v_i(x)\}_{i=1}^\infty$ is a Riesz basis over $L^2(0,\tau)$.

Now let us show that the sequence of the functions $\{u_i(x)\}_{i=1}^{\infty}$ is complete in $L^2(0,\tau)$.
Indeed, for any $f\in L^2(0,\tau)$,
suppose that
$$
\int_0^{\tau}f(x) u_i(x)dx=0,\quad i=0,1,2,\cdots .
$$
In particular, considering $u_0$, we have
$$
\int_0^{\tau}f(x)dx=0,
$$
and therefore,
$$
\int_0^{\tau}f(x) (u_i(x)+1/2)dx=0,\quad i=0,1,2,\cdots .
$$
Integrating by part, we have
$$
\left.\left( \int_x^{\tau}f(s)ds\right)(u_i(x)+1/2)\right|_{x=0}^{x=\tau}-\int^{\tau}_0 v_i(x)\left( \int_x^{\tau}f(s)ds\right)dx=0,
\quad i=1,2,\cdots,
$$
and therefore using the boundary condition $u_i(0)+1/2=0$,
$$
\int^{\tau}_0 v_i(x)\left( \int_x^{\tau}f(s)ds\right)dx=0,\qquad i=1,2,\cdots\\
$$
The sequence $\{v_i(x)\}_{i=1}^\infty$ is complete, hence $\int_x^{\tau}f(s)ds=0$ for $x\in [0,\tau]$. Therefore
$\{u_i(x)\}_{i=1}^\infty$ is complete in $L^2(0,\tau)$.
Since $\{\eta_i u_i(x)\}_{i=0}^\infty$ (where $\eta_i$'s are appropriate normalizing constants) is quadratically close to $\{\cos(\frac{i\pi x}{\tau})\}_{i=0}^\infty$, it follows from Lemma \ref{lem:bar} that
$\{u_i(x)\}_{i=0}^\infty$ is a Riesz basis for $L^2(0,\tau)$.
\end{proof}

\section*{References}
\bibliographystyle{plain}

\begin{thebibliography}{10}

\bibitem{avdonin-kurasov-08}
S.~Avdonin and P.~Kurasov.
\newblock Inverse problems for quantum trees.
\newblock {\em Inverse problems and imaging web}, 1:1--21, 2008.

\bibitem{besicovitch-54}
A.S. Besicovitch.
\newblock {\em Almost periodic functions}.
\newblock Dover, Cambridge, 1954.

\bibitem{Borg-46}
G.~Borg.
\newblock Eine {U}mkehrung des {S}turm-{L}iouvilleschen {E}igenweraufgabe.
\newblock {\em Acta Mathematica}, 78:1--96, 1946.

\bibitem{carlson-98}
R.~Carlson.
\newblock Adjoint and self-adjoint differential operators on graphs.
\newblock {\em Electronic Journal of Differential Equations}, 6:1--10, 1998.

\bibitem{yurko-book}
G.~Freiling and V.~Yurko.
\newblock {\em Inverse {S}turm-{L}iouville problems and their applications}.
\newblock NOVA Science Publishers, 2008.

\bibitem{Gerasimenko-88}
N.I. Gerasimenko.
\newblock The inverse scattering problem on a noncompact graph.
\newblock {\em Theoret. and Math. Phys.}, 75(2):230--240, 1988.

\bibitem{Gerasimenko-87}
N.I. Gerasimenko and B.S. Pavlov.
\newblock A scattering problem on noncompact graphs.
\newblock {\em Theoret. and Math. Phys.}, 74:345--359, 1988.

\bibitem{harmer-02}
M.~Harmer.
\newblock Inverse scattering for the matrix {S}chr\"odinger operator and
  {S}chr\"odinger operator on graphs with general self-adjoint boundary
  conditions.
\newblock {\em ANZIAM J.}, 43:1--8, 2002.

\bibitem{harmer-05}
M.~Harmer.
\newblock Inverse scattering on matrices with boundary conditions.
\newblock {\em J. Phys. A:Math. Gen.}, 38:4875--4885, 2005.

\bibitem{Jaulent-82}
M.~Jaulent.
\newblock The inverse scattering problem for {LCRG} transmission lines.
\newblock {\em J. Math. Phys.}, 23(12):2286--2290, 1982.

\bibitem{kato-book-80}
T.~Kato.
\newblock {\em Perturbation Theory for Linear Operators}.
\newblock Springer, 1980.

\bibitem{Kay-72}
I.~Kay.
\newblock The inverse scattering problems for transmission lines.
\newblock In L.~Collin, editor, {\em Mathematics for Profile inversion-NASA
  Tech. Mem. TM X-62}, volume 150, pages 6--2--6--17, 1972.

\bibitem{kurasov-stenberg}
P.~Kurasov and F.~Stenberg.
\newblock On the inverse scattering problem on branching graphs.
\newblock {\em J. Phys. A}, 35(1), 2002.

\bibitem{marchenko-book}
V.A. Marchenko.
\newblock {\em Sturm-{L}iouville operators and applications}.
\newblock Birkhauser Verlag, 1986.

\bibitem{pivovarchik-07}
V.~Pivovarchik.
\newblock Inverse problem for the {Sturm-Liouville} equation on a star-shaped
  graph.
\newblock {\em Math. Nachr.}, 280:1595--1619, 2007.

\bibitem{roth-55}
K.F. Roth.
\newblock Rational approximations to algebraic numbers.
\newblock {\em Mathematika}, 2:1--20, 1955.

\bibitem{sorine-riccati}
M.~Sorine and P.~Winternitz.
\newblock {Superposition Laws for Solutions of Differential Matrix Riccati
  Equations Arising in Control Theory}.
\newblock {\em IEEE Trans. Automat. Control}, 30(3):266--272, 1985.

\bibitem{yurko-05}
V.~Yurko.
\newblock Inverse spectral problems for sturm-liouville operators on graphs.
\newblock {\em Inverse problems}, 21(3):1075--1086, 2005.

\bibitem{yurko-07}
V.~Yurko.
\newblock An inverse problem for higher order differential operators on
  star-type graphs.
\newblock {\em Inverse problems}, 23(3):893--903, 2007.

\bibitem{yurko-08}
V.~Yurko.
\newblock Inverse problems for {S}turm-{L}iouville operators on graphs with a
  cycle.
\newblock {\em Oper. Matrices 2}, 4:543--553, 2008.

\end{thebibliography}

\end{document}